\documentclass[journal]{IEEEtran}

\usepackage{amsmath,amssymb,amsfonts}
\usepackage{amsthm}
\usepackage{algorithm}
\usepackage{algorithmic}
\usepackage{graphicx}
\usepackage{cite}
\usepackage{url}
\usepackage{booktabs}
\usepackage{multirow}
\usepackage{color}
\usepackage{bm}
\usepackage{bbm}
\usepackage{mathrsfs}
\usepackage{pifont} 

\newcommand{\cmark}{\ding{51}}

\newcommand{\xmark}{\ding{55}}

\newtheorem{theorem}{Theorem}

\newtheorem{lemma}{Lemma}

\newtheorem{proposition}{Proposition}

\newtheorem{corollary}{Corollary}

\newtheorem{definition}{Definition}

\newtheorem{remark}{Remark}

\newtheorem{assumption}{Assumption}

\newcommand{\CB}{C^{B}}

\newcommand{\CBmin}{C^{B}_{min}}

\newcommand{\CBmax}{C^{B}_{max}}

\newcommand{\Pmax}{P_{max}}

\newcommand{\Pharvest}{P^{harvest}}

\newcommand{\calM}{\mathcal{M}}

\newcommand{\calG}{\mathcal{G}}

\newcommand{\calT}{\mathcal{T}}

\newcommand{\bfP}{\mathbf{P}}

\newcommand{\bfR}{\mathbf{R}}

\newcommand{\Rmn}{R_{m,n}^{t}}

\newcommand{\WassDist}{W_{1}}

\begin{document}

\title{Hierarchical Battery-Aware Game Algorithm 
for ISL Power Allocation in LEO Mega-Constellations}

\author{Kangkang Sun,
        Jianhua Li,~\IEEEmembership{Senior Member,~IEEE,}
        Xiuzhen Chen,~\IEEEmembership{Member,~IEEE,}\\
        Jianyong Zheng,~\IEEEmembership{Member,~IEEE,}
        Minyi Guo,~\IEEEmembership{Fellow,~IEEE}
  \thanks{Kangkang Sun, Jianhua Li, Xiuzhen Chen and Minyi Guo are
          with the Shanghai Key Laboratory of Integrated Administration
          Technologies for Information Security, School of Computer Science,
          Shanghai Jiao Tong University, Shanghai 200240, China
          (e-mail: szpsunkk@sjtu.edu.cn; lijh888@sjtu.edu.cn;
          xzchen@sjtu.edu.cn; guo-my@cs.sjtu.edu.cn).}
  \thanks{Jianyong Zheng is with the School of Future Technology, Shanghai University, Shanghai 200444, China (e-mail: zhengjy@shu.edu.cn).}
  \thanks{\textit{Corresponding author: Jianhua Li}.}
  \thanks{This work has been submitted to the IEEE for possible publication. Copyright may be transferred without notice, after which this version may no longer be accessible.}
}

\maketitle

\begin{abstract}
Sustaining high inter-satellite link (ISL) throughput under intermittent solar harvesting is a fundamental challenge for LEO mega-constellations. Existing works impose static power ceilings that ignore real-time battery state and comprehensive onboard power budgets, causing eclipse-period energy crises. Learning-based approaches capture battery dynamics but lack equilibrium guarantees and do not scale beyond small constellations. We propose the \textbf{Hierarchical Battery-Aware Game (HBAG)} algorithm, a unified game-theoretic framework for ISL power allocation that operates identically across finite and mega-constellation regimes. For finite constellations, HBAG converges to a unique variational equilibrium; as constellation size grows, the same distributed update rule converges to the Mean Field Game (MFG) equilibrium without algorithm redesign. Comprehensive experiments on Starlink Shell~A ($M=172$, $\theta=0.38$) show that HBAG achieves \textbf{100\% energy sustainability rate} (ESR) in all 10 independent runs, representing a \textbf{+87.4\%} gain over the traditional static-power baseline (SATFLOW-L, ESR\,=\,12.6\%). At the same time, HBAG reduces the flow violation ratio by \textbf{78.3\%} to 7.62\% (below the 10\% industry tolerance). HBAG further maintains ESR $\geq 93.4\%$ across eclipse fractions $\theta \in [0,\,0.6]$ and scales linearly to 5{,}000 satellites with less than 75\,ms per-slot runtime, confirming deployment feasibility at full Starlink scale.
\end{abstract}


\section{Introduction}
\label{sec:intro}
\IEEEPARstart{L}{ow} Earth Orbit (LEO) satellite mega-constellations have emerged as a transformative paradigm for global connectivity, enabling broadband Internet access across oceans, polar regions, and underserved rural areas where terrestrial infrastructure is absent or economically prohibitive. Central to their operation are inter-satellite links (ISLs), namely laser or RF channels supporting multi-Gbps throughput between satellites, which eliminate dependence on ground gateways and enable seamless end-to-end routing across orbital shells with latencies competitive to terrestrial fiber. As constellation scales continue to grow from hundreds to thousands of satellites, sustaining ISL throughput under the stringent energy constraints of space becomes a fundamental engineering challenge \cite{kodheli2020survey,al2022survey}.

The economic viability of mega-constellations hinges on maximizing ISL availability over the satellite lifespan while minimizing operational expenditures. Any premature battery failure requiring satellite replacement incurs substantial hardware and launch costs \cite{heydarishahreza2024spectrum}. Each satellite harvests solar energy via deployable panels and stores it in onboard lithium-ion batteries. This harvested energy must be shared among ISL laser terminals, attitude control, onboard processing, and thermal management subsystems. Unlike terrestrial base stations drawing from the grid with effectively unlimited supply, LEO satellites operate under a closed energy budget: every joule spent on ISL transmission depletes the battery reserve available for eclipse periods, and repeated deep discharge cycles accelerate irreversible capacity degradation \cite{kodheli2020survey}. This closed-loop coupling between ISL power allocation and battery longevity motivates the need for energy-aware frameworks that go beyond static power ceilings.

\subsection{Motivation: The Eclipse-Period Energy Crisis}
\label{subsec:motivation}

LEO satellites face a fundamental energy intermittency challenge driven by orbital mechanics \cite{kodheli2020survey}. Each orbital period comprises an illumination phase, during which solar panels harvest energy and recharge the battery, and an eclipse phase, during which solar power drops to zero and the satellite must rely entirely on stored energy to sustain all onboard subsystems, including ISL terminals. The eclipse phase occupies a significant fraction of each orbit for typical constellation inclinations, and the aggregate power draw from ISL transmission, attitude control, processing, and thermal management can consume a substantial portion of the battery reserve within a single eclipse window \cite{radhakrishnan2016survey}.

The critical vulnerability arises from the coupling between consecutive orbital cycles: if the illumination phase fails to fully replenish the energy consumed during eclipse, due to suboptimal power allocation prioritizing ISL throughput over battery replenishment, successive eclipse cycles cause cumulative depletion. Left unchecked, this cascading energy deficit eventually triggers catastrophic battery exhaustion and involuntary transmitter shutdown, severing ISL connectivity and disrupting end-to-end routing across the constellation \cite{al2022survey}.
This eclipse-period energy crisis is invisible to existing static-power frameworks, which enforce fixed power ceilings regardless of real-time battery state or orbital phase \cite{cen2024satflow}. It exposes three compounding gaps in the literature that no single prior work addresses simultaneously in Table \ref{tab:gap_analysis}.

\begin{table*}[t]
\centering
\caption{Gap Analysis: Limitations of Existing Approaches for ISL Power Allocation Under Eclipse Constraints. The FVR gap between HBAG (7.62\%) and SMFG-adapted (0.15\%) reflects the \emph{price of distribution}; see Section~\ref{subsubsec:tradeoff}.}
\label{tab:gap_analysis}
\small
\begin{tabular}{p{2.3cm}p{3.3cm}p{3.2cm}p{3.4cm}p{3.2cm}}
\toprule
\textbf{Approach} & \textbf{Category} & \textbf{Battery Dynamics} & \textbf{Equilibrium Guarantees} & \textbf{Scalability ($M > 1000$)} \\
\midrule
SATFLOW~\cite{cen2024satflow} & Optimization & \xmark\ Static $P_{\max}$ & \xmark\ No equilibrium & \cmark\ $\mathcal{O}(n_t)$ \\
MAAC-IILP~\cite{li2023energy} & Multi-agent RL & \cmark\ Dynamic $C_m^B(t)$ & \mbox{\xmark\ No convergence proof} & \xmark\ Trained on $M=66$ \\
DeepISL~\cite{li2023deepisl} & Deep RL & \cmark\ Dynamic $C_m^B(t)$ & \xmark\ No equilibrium & \xmark\ $\mathcal{O}(M^2)$ training \\
FlexD~\cite{lokugama2026flexd} & Analytical & \xmark\ Battery-agnostic & N/A (heuristic) & \cmark\ $\mathcal{O}(1)$ \\
SMFG~\cite{wang2023smfg} & MFG & Partial (offloading only) & \cmark\ Stackelberg eq. & \cmark\ $\mathcal{O}(M \log M)$ \\
\midrule
\textbf{This work} & \textbf{Potential game + MFG} & \cmark\ \textbf{Dynamic} $P_{m,t}^{max}$ & \cmark\ \textbf{Unique VE}$^{\dagger}$ & \cmark\ $\mathcal{O}(n_t)^{\ddagger}$ \\
\bottomrule
\multicolumn{5}{l}{\footnotesize $^{\dagger}$Under quasi-static assumption (Remark~\ref{rem:timescale}) and interference-free ISLs (Assumption~\ref{assm:interference_free}).} \\
\multicolumn{5}{l}{\footnotesize $^{\ddagger}$FVR (7.62\%) exceeds SMFG-adapted (0.15\%) due to distributed Nash equilibrium trade-off; see Table~\ref{tab:exp1_results}.} \\
\end{tabular}
\end{table*}

\subsection{Research Gaps and Challenges}
\label{subsec:gaps}

\textit{Gap~1: Static power budgets ignore eclipse-driven battery dynamics.}
All major ISL planning frameworks \cite{cen2024satflow,deng2020ultra,gupta2025energyefficient} enforce fixed power ceilings irrespective of real-time battery state or eclipse phase. This causes two critical failures: 1) During eclipse, enforcing a static maximum power overdrafts the battery faster than it can be replenished in the next illumination window, eventually triggering involuntary transmitter shutdown, a failure mode that remains entirely invisible to the optimizer's cost function. 2) During illumination, satellites recharge to full capacity but often fail to reserve an adequate energy margin for the upcoming eclipse, causing cascading depletion in subsequent orbital cycles.
Learning-based methods MAAC-IILP \cite{li2023energy} and DeepISL \cite{li2023deepisl} partially address battery dynamics by incorporating battery state into the reinforcement learning observation space, achieving modest energy savings on small-scale constellations. However, both methods exhibit three fundamental scalability barriers: (a) Training overhead grows super-linearly with constellation size, rendering direct extrapolation to mega-constellations computationally infeasible on current hardware; (b) neither method provides convergence guarantees or equilibrium characterization, creating fundamental uncertainty about solution quality at large scales; (c) policies trained at one constellation scale fail to generalize across deployment sizes, necessitating costly retraining for every scale change.

\textit{Gap~2: Game-theoretic approaches target the wrong interaction layer.}
While game theory has been applied to satellite topology optimization~\cite{han2024topology}, switch migration~\cite{yan2024switchmigration}, and integrated scheduling~\cite{fu2013transmission}, none address \emph{horizontal satellite-to-satellite ISL power competition} under battery constraints. The sole prior mean field game (MFG) work for LEO energy management, SMFG~\cite{wang2023smfg}, models a \emph{vertical} Stackelberg game between satellites (leaders) and ground users (followers) for data offloading pricing. SMFG's Stackelberg hierarchy is fundamentally mismatched with our setting \cite{wang2023smfg}. The most critical issue is the vertical leader-follower assumption, which presupposes that satellites can commit to pricing strategies before peers respond—a premise that breaks down entirely when satellites compete symmetrically for ISL bandwidth. A secondary but equally important gap concerns the role of battery state: in SMFG, energy constrains offloading capacity but never enters the ISL competition itself. These two issues together mean that adapting SMFG requires not just parameter tuning but architectural redesign.


\textit{Gap~3: Lack of unified scalability from finite to mega-constellation regimes.}
No existing work provides an algorithm that \emph{operates identically} and \emph{provides convergence guarantees} across constellation sizes from $M \sim 100$ (regional systems like Iridium) to $M \sim 10{,}000$ (proposed Amazon Kuiper scale). Optimization-based methods scale linearly but ignore energy dynamics; learning-based methods capture energy but require retraining; game-theoretic methods provide equilibrium theory but have been applied only to vertical interactions or small-scale scenarios. The fundamental challenge is: \textbf{\textit{Can we design a single distributed algorithm that converges to an exact equilibrium for finite $M$ and seamlessly approximates the mean field equilibrium as $M \to \infty$, with quantifiable convergence rates in both regimes?}}
\smallskip

\subsection{Contributions and Technical Novelty}
\label{subsec:contributions}

This paper proposes the Hierarchical Battery-Aware Game (HBAG) algorithm, a unified game-theoretic framework for ISL power allocation that addresses all three gaps identified above. Our main contributions are:
\begin{enumerate}
    \item \textit{Battery-aware exact potential game with unique variational equilibrium (Theorem~\ref{thm:potential}).} The ISL power allocation is fomulated as an exact potential game with a battery-state-dependent penalty that diverges as charge approaches the safety floor. We prove that the potential function is strictly jointly concave, guaranteeing a unique variational equilibrium despite the time-varying feasible set induced by eclipse-driven dynamic power bounds.
    \item \textit{Unified hierarchical algorithm with dual convergence guarantees (Theorems~\ref{thm:convergence} and~\ref{thm:mfg}).} A single distributed algorithm is proposed that converges to the unique VE at rate $\mathcal{O}(1/\sqrt{k})$ for finite constellations and approximates the mean field equilibrium at rate $\mathcal{O}(M^{-1/4})$ as $M \to \infty$, without any algorithmic redesign.
    \item \textit{Comprehensive empirical validation on Starlink Shell~A.} The Energy Sustainability Rate (ESR) metric is introduced and shown through seven experiments that HBAG achieves 100\% ESR (87.4 pp improvement over SATFLOW), scales linearly to 5{,}000 satellites, and maintains flow violation ratio below the 10\% industry tolerance.
\end{enumerate}

The remainder of this paper is organized as follows. Section~\ref{sec:related} reviews related work. Section~\ref{sec:model} presents the system model. Section~\ref{sec:game} formulates the game and proves equilibrium results. Section~\ref{sec:algorithm} presents the hierarchical algorithm and unified convergence properties. Section~\ref{sec:asymptotic} provides finite-to-mean-field asymptotic analysis. Sections~\ref{sec:simulation} and~\ref{sec:conclusion} present simulation results and conclude the paper, respectively.

\section{Related Work}
\label{sec:related}

This section reviews related work in three categories aligned with the research gaps identified in Section~\ref{subsec:gaps}: ISL network planning and power allocation (Section~\ref{subsec:related_planning}), game-theoretic approaches to satellite networks (Section~\ref{subsec:related_game}), and mean field games for large-scale systems (Section~\ref{subsec:related_mfg}).
\subsection{ISL Network Planning and Energy-Aware Power Allocation}
\label{subsec:related_planning}

\textit{Optimization-based frameworks.}
Early topology-centric works~\cite{leyva2021connectivity,di2019ultra} focus on connectivity under geometric constraints but overlook energy entirely. SATFLOW~\cite{cen2024satflow} represents the state-of-the-art, jointly optimizing topology re-establishment and per-link power allocation via hierarchical decomposition. However, its power constraint $P_{m,n}^t \leq a_{m,n}^t \Pmax$ is \emph{battery-agnostic}: $\Pmax$ is fixed irrespective of $\CB_m(t)$ or $\phi_m(t)$. Deng et al.~\cite{deng2020ultra} and Gupta et al.~\cite{gupta2025energyefficient} similarly treat power as a static resource or soft cost without hard eclipse-driven battery constraints. FlexD~\cite{lokugama2026flexd} achieves 30\% energy efficiency gains via flexible duplex ISL operation but maximizes instantaneous throughput without battery sufficiency checks. Energy-aware routing works~\cite{yang2016energyrouting,marchese2020ecgr,tang2019multipath} address different problem layers (routing vs.\ power control) and do not provide equilibrium guarantees. Radhakrishnan et al.~\cite{radhakrishnan2016survey} provide a comprehensive survey of ISL physical-layer parameters for small satellite systems.


\textit{Learning-based methods with battery dynamics.}
MAAC-IILP~\cite{li2023energy} and DeepISL~\cite{li2023deepisl} are the first to incorporate physically grounded battery dynamics into ISL planning, augmenting the RL state space with $\CB_m(t)$ and $\phi_m(t)$. Critical limitations: \emph{(i)} \textit{Scalability:} Training overhead grows super-linearly with constellation size, rendering direct extrapolation from small-scale testbeds to mega-constellations computationally infeasible on current GPU clusters. \emph{(ii)} \textit{Generalization:} Policies trained on small constellations suffer significant performance degradation when directly applied to larger deployments, confirming poor cross-scale transfer. \emph{(iii)} \textit{Theory:} Neither method provides convergence guarantees or equilibrium characterization.

All planning frameworks except learning-based methods ignore eclipse-driven battery dynamics. Our potential game formulation provides: \emph{(i)} battery-state-dependent constraints via $P_{m,t}^{max}$; \emph{(ii)} equilibrium guarantees via strict concavity of $\Phi$; and \emph{(iii)} linear scalability via mean field reduction.

\subsection{Game Theory for Satellite Networks and Wireless Power Control}
\label{subsec:related_game}

\textit{Game theory in satellite systems.}
Game-theoretic modeling of satellite networks has gained traction for diverse problems. Jiang et al.~\cite{jiang2024gametheory} survey recent applications, identifying resource allocation and topology optimization as key domains. Han et al.~\cite{han2024topology} apply potential game theory to VLEO-LEO hybrid network topology optimization, formulating satellite selection as a congestion game and proving NE existence via Rosenthal's theorem. However, their utility does not incorporate battery state dynamics---energy is implicitly assumed unlimited. Yan et al.~\cite{yan2024switchmigration} study switch migration and integrated scheduling via game-theoretic mechanisms (Stackelberg and matching games, respectively), confirming the broader applicability of game theory to satellite systems but not addressing ISL power allocation under energy constraints. Jiao et al.~\cite{jiao2020noma} study network utility maximization for NOMA-based satellite IoT, modeling energy as a soft cost in the objective (penalty weight $\beta > 0$) rather than a hard battery constraint tied to eclipse dynamics.

\textit{Mean field games for LEO energy management.}
SMFG~\cite{wang2023smfg} is the \emph{only} prior work applying MFG to LEO satellite energy management, formulating a Stackelberg Mean Field Game for data offloading pricing. The satellite energy state $\tilde{E}^j(t)$ follows $\mathrm{d}\tilde{E}^j = (-\tilde{p}^j + \tilde{\alpha}^j)\mathrm{d}t + \sigma \mathrm{d}W$ (equation~(6) in~\cite{wang2023smfg}), where $\tilde{p}^j$ is offloading power and $\tilde{\alpha}^j$ is harvesting rate. Three structural differences prevent SMFG from addressing our problem:
\begin{enumerate}
    \item \textbf{Vertical vs.\ horizontal interaction:} SMFG models satellites (leaders) setting prices to maximize revenue, and ground users (followers) choosing offloading strategies. This Stackelberg hierarchy assumes satellites can commit to strategies before users respond, inappropriate for symmetric peer satellites competing for ISL resources. Our game requires Nash (or variational) equilibrium among equals.
    \item \textbf{Energy state role:} In SMFG, $\tilde{E}^j(t)$ constrains offloading capacity ($\tilde{p}^j \leq \tilde{E}^j / \Delta t$) but does not enter the ISL power allocation among satellites. ISL competition is never modeled. Our formulation makes $C_m^B(t)$ a first-class state variable determining feasible power $P_{m,t}^{max}$.
    \item \textbf{Convergence rate:} SMFG provides no analysis of finite-to-MFE convergence rate, leaving approximation quality at practical large-scale $M$ unknown. We establish $\mathcal{O}(M^{-1/4})$ rate (Theorem~\ref{thm:mfg}).
\end{enumerate}

\textit{Potential games for wireless power control.}
Scutari et al.~\cite{scutari2008mimo} formulate MIMO interference channel power control as an exact potential game, proving NE existence under log-determinant utility. The foundational potential game theory of Monderer and Shapley~\cite{monderer1996potential} and the broader treatment in~\cite{lasaulce2011gametheory} provide the theoretical basis for our formulation. However, all terrestrial potential game frameworks assume interference-coupled utilities with \emph{fixed} feasible sets, whereas our game features eclipse-driven battery-state-dependent $P_{m,t}^{max}$ that varies the feasible set. We bridge this gap by proving strict concavity of $\Phi$ under the quasi-static approximation, guaranteeing VE uniqueness despite time-varying constraints.

\textit{Variational equilibrium for games with shared constraints.}
When players face shared constraints (e.g., flow conservation in ISL networks), standard NE is inappropriate because individual feasible sets are not separable. Facchinei and Kanzow~\cite{facchinei2007vi} formalize variational equilibrium (VE) via variational inequality: $\mathbf{R}^*$ is a VE if $\sum_m \langle 
\nabla u_m(\mathbf{R}^*), \mathbf{R}_m - \mathbf{R}_m^* \rangle \leq 0$ for all $\mathbf{R} \in \mathcal{F}$ (shared feasible set). Our Lagrangian decomposition (Section~\ref{sec:game}) implicitly computes a VE: dualizing flow conservation yields separable sub-problems whose NE coincides with the VE of the original constrained game when dual variables satisfy optimality.

Prior game-theoretic works for satellites ignore battery state dynamics (Han~\cite{han2024topology}, Yan~\cite{yan2024switchmigration}) or model the wrong interaction layer (SMFG's vertical Stackelberg vs.\ our horizontal Nash). Terrestrial potential games~\cite{scutari2008mimo} provide the theoretical foundation but assume fixed feasible sets, whereas our game features battery-state-dependent $P_{m,t}^{max}$. We bridge this gap by proving strict concavity of $\Phi$ under the quasi-static approximation, guaranteeing VE uniqueness despite time-varying constraints.

\subsection{Mean Field Games for Large-Scale Communication Systems}
\label{subsec:related_mfg}

\textit{MFG foundations and LEO applications.}
MFG theory~\cite{lasry2007mfg,huang2012socialoptima} provides a principled framework for large-population strategic interactions, characterizing equilibria via a coupled HJB--FPK system. Applications to communication networks include uplink power control~\cite{huang2012socialoptima} and energy-harvesting systems with stochastic arrivals~\cite{tembine2014risksensitive}. The key distinction of our setting is that LEO satellite battery dynamics are \emph{deterministic} (driven by orbital mechanics), yielding a pure-advection FPK rather than a diffusion--advection equation. Most communication MFG works do not establish explicit finite-to-MFE convergence rates, leaving approximation quality at practical $M \sim 10^3$--$10^4$ unquantified; we establish $\mathcal{O}(M^{-1/4})$ rate (Theorem~\ref{thm:mfg}).
Fournier and Guillin~\cite{fournier2015wasserstein} establish Wasserstein-1 convergence rates for empirical measures: if $\{X_i\}_{i=1}^M$ are i.i.d.\ samples from $\mu^*$ with finite $p$-th moment ($p \geq 1$), then $\mathbb{E}[W_1(\mu^M, \mu^*)] \leq C_\mu / M^{1/2}$ for dimension $d=1$ (battery state in our case). This result enables quantitative analysis of finite-to-MFE approximation quality. However, most MFG works in communications do not establish explicit convergence rates, leaving approximation quality at practical system sizes unquantified.

\section{System Model}
\label{sec:model}

\begin{table}[t]
\centering
\caption{Key Notation}
\label{tab:notation}
\renewcommand{\arraystretch}{1.1}
\begin{tabular}{ll}
\toprule
\textbf{Symbol} & \textbf{Description} \\
\midrule
$M$, $\calM$ & Number of satellites; index set $\{0,\ldots,M-1\}$ \\
$R_{m,n}^t$ & Allocated data rate on ISL $(m,n)$ at slot $t$ \\
$P_{m,n}^t$ & Transmit power on ISL $(m,n)$ at slot $t$ \\
$P_{m,t}^{max}$ & Dynamic power upper bound for satellite $m$ at $t$ \\
$\CB_m(t)$ & Battery state-of-charge (kJ) of satellite $m$ at $t$ \\
$\CBmin$, $\CBmax$ & Minimum safe charge (40 kJ); maximum capacity (400 kJ) \\
$\phi_m(t)$ & Illumination indicator ($1$ = illuminated, $0$ = eclipse) \\
$\delta_{\mathrm{ecl}}(t)$ & Remaining eclipse duration (s) at slot $t$ \\
$\kappa_{m,n}^t$ & Path-loss/hardware coefficient for ISL $(m,n)$ at $t$ \\
$\lambda_m(t)$ & Energy sensitivity coefficient of satellite $m$ at $t$ \\
$\epsilon$ & Battery safety margin parameter\\
$\Phi(\bfR)$ & Exact potential function of the penalized game \\
$\mu(\CB, t)$ & Battery-state distribution in the mean field limit \\
$\bfR^*$ & Unique variational equilibrium strategy profile \\
$\alpha_\Phi$ & Strong concavity modulus of $\Phi$ \\
$L_U$ & Lipschitz constant of utility in the mean field \\
\bottomrule
\end{tabular}
\end{table}

This section presents the system model for battery-aware ISL power allocation. We first describe the constellation architecture and ISL topology (Section~\ref{subsec:arch}), then introduce the communication model governing link capacity (Section~\ref{subsec:comm}), the energy model capturing solar harvesting, battery dynamics, and dynamic power bounds (Section~\ref{subsec:energy}), and finally formulate the optimization problem (Section~\ref{subsec:formulation}). Figure~\ref{fig:scenario} illustrates the operational scenario: satellites in illumination harvest solar power and maintain high-power ISLs, while satellites entering eclipse progressively reduce transmit power as battery state decreases.

\begin{figure}[!t]
    \centering
    \includegraphics[width=0.5\textwidth]{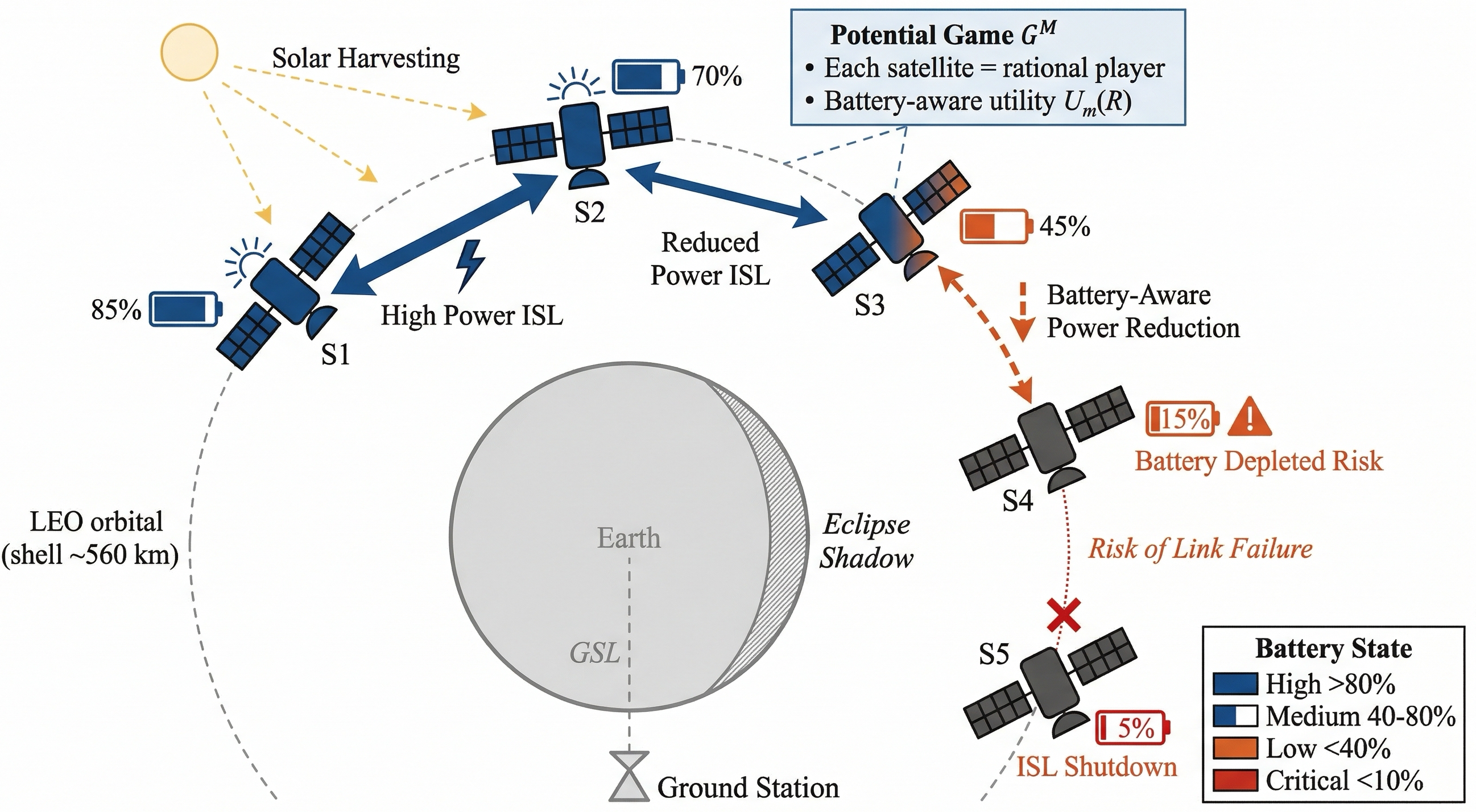}
    \caption{Scenario illustration of battery-aware ISL
    operation in LEO: illuminated satellites harvest solar
    energy and maintain high-power links, while eclipse-side
    satellites reduce power or shut down as state-of-charge decreases.}
    \label{fig:scenario}
\end{figure}

\subsection{Constellation Architecture}
\label{subsec:arch}

We consider $M$ LEO satellites evenly distributed in $N_{orbit}$ orbital planes at altitude $h$ and inclination $\iota$. The satellite index set is $\calM = \{0,1,\ldots,M-1\}$.
Each satellite is equipped with four laser terminals supporting up to four full-duplex ISLs: two permanent intra-plane ISLs and two dynamic inter-plane ISLs~\cite{cen2024satflow}. Let $\mathbf{A}^t \in \{0,1\}^{M \times M}$ denote the ISL adjacency matrix at time $t$, where $a_{m,n}^t = 1$ if an ISL is established from satellite $m$ to $n$.
Terminal power is adjusted every $D_e$ seconds over the time-slot set $\mathcal{T}_e = \{t_0, t_1, \ldots, t_{n_t-1}\}$.

While operational Starlink satellites employ optical ISLs, we adopt an RF ISL model at Ka-band (26~GHz) following the validated setup in~\cite{cen2024satflow}, as the Shannon capacity framework~\eqref{eq:capacity} is well established for RF links.
We set ISL bandwidth $B = 500$~MHz in Table~\ref{tab:parameters}; Ka-band regulatory and payload allocations support wideband RF ISL operation, and the same value is used in the Shannon evaluation for reproducibility.
The core battery-aware game formulation is agnostic to the specific channel model: it only requires the power-rate mapping $P_{m,n}^t(R_{m,n}^t)$ to be convex and increasing. For coherent optical ISLs with homodyne detection, the mapping becomes $P_{m,n}^t(R) = (h\nu / \eta_{\mathrm{det}}) \cdot (2^{R/B_{\mathrm{opt}}} - 1)$, where $h\nu$ is photon energy and $\eta_{\mathrm{det}}$ is detector quantum efficiency. This mapping is convex and increasing in $R$, so $\Phi$ retains strict joint concavity with a curvature constant $\kappa_{m,n}^{\mathrm{opt}}$ replacing $\kappa_{m,n}^t$. 

\subsection{Communication Model}
\label{subsec:comm}

The free-space path loss from satellite $m$ to $n$ at time $t$ is~\cite{cen2024satflow}:
\begin{equation}
    L_{m,n}^t = \left(\frac{4\pi d_{m,n}^t f}{c}
    \right)^2,
    \label{eq:pathloss}
\end{equation}
where $d_{m,n}^t$ is the Euclidean distance, $f$ is carrier frequency, and $c$ is the speed of light. The channel capacity is:
\begin{equation}
    C_{m,n}^t = B\log_2\!\left(1 +
    \frac{p_m^t G_m G_n}{k_B \tau B L_{m,n}^t}
    \right).
    \label{eq:capacity}
\end{equation}
The required power for allocated data rate $R_{m,n}^t = \sum_{\omega} Y_{m,n}^{\omega,t}$ is:
\begin{equation}
    P_{m,n}^t(R_{m,n}^t) = \kappa_{m,n}^t
    \!\left(2^{R_{m,n}^t/B} - 1\right),
    \label{eq:power_rate}
\end{equation}
where $\kappa_{m,n}^t = a_{m,n}^t
\frac{k_B\tau B}{G_m G_n}
\!\left(\frac{4\pi d_{m,n}^t f}{c}\right)^{\!2}$.

\subsection{Energy Model}
\label{subsec:energy}

The instantaneous solar harvesting power of satellite $m$ at time $t$ is
\begin{equation}
    \Pharvest_{m,t} = \phi_m(t)\cdot\eta\cdot
    \gamma'\cdot A\cdot\sin\sigma_m(t),
    \label{eq:harvest}
\end{equation}
as modeled in~\cite{li2023energy,li2023deepisl}, where $\phi_m(t)\in\{0,1\}$ is the illumination indicator ($\phi_m(t)=0$ during eclipse), $\eta$ is solar panel efficiency, $\gamma'$ is solar irradiance per unit area, $A$ is solar panel area, and $\sigma_m(t)$ is the angle between the panel and sunlight.

The net energy change per time slot and the resulting battery update are
\begin{equation}
    \begin{aligned}
        \Delta E_m^t & = E_m^{harvest}(t) -
        E_m^{consume}(t), \\
        E_m^{consume}(t)&= D_e\left(\sum_{n}P_{m,n}^{t} + P_m^{base}\right),
    \end{aligned}
    \label{eq:deltaE}
\end{equation}
\begin{equation}
    \CB_m(t{+}1) = \begin{cases}
        \min\!\left(\CB_m(t)+\Delta E_m^t,\,
        \CBmax\right),
        & \Delta E_m^t > 0 \\
        \max\!\left(\CB_m(t)+\Delta E_m^t,\,0
        \right),
        & \Delta E_m^t \leq 0
    \end{cases}
    \label{eq:battery}
\end{equation}
where $P_m^{base}$ denotes the baseline power consumption from non-ISL subsystems: attitude control system (ACS), onboard processors, and thermal management. This comprehensive energy model ensures realistic battery depletion dynamics consistent with operational LEO satellites~\cite{li2023energy,li2023deepisl}.

The linear battery model in equations~\eqref{eq:deltaE}--\eqref{eq:battery} adopts three common simplifications: \emph{(i)}~unity charge/discharge efficiency \cite{radhakrishnan2016survey}, and affine corrections preserve convexity in $P_{m,n}^t$; \emph{(ii)}~constant $\CBmax$ over lifespan. 

With efficiencies $\eta_c, \eta_d < 1$, the battery dynamics become:
\begin{equation}
    \CB_m(t{+}1) = \begin{cases}
    \min(\CB_m(t) + \eta_c \Delta E_m^t, \CBmax), 
    & \Delta E_m^t > 0 \\
    \max(\CB_m(t) + \Delta E_m^t / \eta_d, 0), 
    & \Delta E_m^t \leq 0.
    \end{cases}
\end{equation}


To prevent battery depletion, we impose a time-varying power ceiling that accounts for both the current battery state and the remaining eclipse duration:
\begin{equation}
    P_{m,t}^{max} = \begin{cases}
        \min\!\left(\dfrac{\CB_m(t)}{\delta_{\mathrm{ecl}}(t)},\,
        \Pmax\right),
        & \phi_m(t)=0\\[6pt]
        \min\!\left(\Pharvest_{m,t}+
        \dfrac{\CB_m(t)}{\delta_{\mathrm{ecl}}(t)},\,\Pmax\right),
        & \phi_m(t)=1
    \end{cases}
    \label{eq:dynpower}
\end{equation}
where $\delta_{\mathrm{ecl}}(t)>0$ denotes the remaining eclipse duration (in seconds): during eclipse ($\phi_m(t)=0$), it equals the time until illumination resumes; during illumination ($\phi_m(t)=1$), it equals the duration of the upcoming eclipse, so that the satellite reserves sufficient energy to survive the next dark period \cite{cen2024satflow, li2023deepisl}.


\subsection{Problem Formulation}
\label{subsec:formulation}

Our objective is to minimize total ISL energy consumption subject to energy sustainability constraints. We replace the static power constraint $P_{m,n}^t \leq a_{m,n}^t \Pmax$ from~\cite{cen2024satflow} with the eclipse-aware dynamic bound~\eqref{eq:dynpower} and add explicit battery floor constraints:

\begin{subequations}
\label{prob:main}
\begin{align}
    \min_{\bfR,\,\mathbf{A}} \quad &
    \alpha D_e \sum_{t\in\mathcal{T}_e}
    \sum_{m,n\in\calM} P_{m,n}^t\!\left(\Rmn\right)
    + \beta\sum_{t\in\mathcal{T}_i} S^t
    \label{eq:obj}\\
    \text{s.t.} \quad &
    P_{m,n}^t\!\left(\Rmn\right) \leq
    a_{m,n}^t \cdot P_{m,t}^{max},
    \;\forall m,n,t
    \label{eq:c1}\\
    & \sum_{n} Y_{m,n}^{\omega,t} -
    \sum_{n} Y_{n,m}^{\omega,t} = [b^\omega]_m,
    \;\forall\omega,m,t
    \label{eq:c2}\\
    & 0\leq Y_{m,n}^{\omega,t} \leq
    a_{m,n}^t d_\omega,
    \;\forall\omega,m,n,t
    \label{eq:c3}\\
    & \CB_m(t) \geq \CBmin, \quad\forall m,t
    \label{eq:c4}\\
    & \mathbf{A}^t \in \mathcal{A}^t,
    \;\forall t,
    \label{eq:c5}
\end{align}
\end{subequations}
where constraint~\eqref{eq:c4} is the \emph{energy sustainability constraint}. Problem~\eqref{prob:main} remains NP-hard because it strictly subsumes the fixed-charge network flow structure identified. Constraint~\eqref{eq:c4} serves as the \emph{hard} battery floor defining feasibility. In the game formulation (Section~\ref{sec:game}), the battery penalty term in utility~\eqref{eq:utility} provides a complementary \emph{soft} surrogate: as $\CB_m(t)\downarrow\CBmin$, the penalty diverges, making overdraft individually irrational. These two mechanisms are active simultaneously in Algorithm~\ref{alg:hbag}: the dynamic bound~\eqref{eq:dynpower} guarantees feasibility, while the soft penalty guides convergence toward energy-efficient solutions within the feasible set---their complementary roles are quantified in the ablation study (Table~\ref{tab:ablation}).

\section{Game Formulation and Equilibrium Analysis}
\label{sec:game}

This section transforms the system model of Section~\ref{sec:model} into a game-theoretic framework for ISL power allocation. We first decompose the optimization problem~\eqref{prob:main} into a multi-player game and adopt the variational equilibrium (VE) concept to handle the shared flow conservation constraint (Section~\ref{subsec:game_decomp}). We then define the penalized game with its strategy space and battery-aware utility function (Section~\ref{subsec:penalized_game}), and construct an exact potential function that captures the battery-state-dependent cost structure while preserving additive separability under Assumption~\ref{assm:interference_free} (Section~\ref{subsec:potential}). Finally, we prove that the penalized game admits a unique variational equilibrium guaranteed by the strict joint concavity of the potential function, and derive the associated corollary and deviation bound proposition (Section~\ref{subsec:theory}).

\subsection{Game Decomposition and Equilibrium Concept}
\label{subsec:game_decomp}

For a fixed ISL topology $\mathbf{A}^t$, we consider the power allocation sub-problem. The optimization problem~\eqref{prob:main} (with $\mathbf{A}^t$ fixed) involves both \emph{private} constraints (constraint~\eqref{eq:c1}: per-satellite power caps) and a \emph{shared} constraint (constraint~\eqref{eq:c2}: flow conservation, coupling all satellites' variables). The presence of shared constraints means that the standard Nash Equilibrium (NE) concept, which assumes separable strategy sets, does not directly apply. We adopt the \emph{variational equilibrium} (VE) concept~\cite{facchinei2007vi} appropriate for games with shared constraints. First, we give the following assumption:

\begin{assumption}[Interference-Free ISLs]
    \label{assm:interference_free}
    Each satellite is equipped with narrow-beam laser terminals with precise beam alignment, enabling communication in an interference-free environment~\cite{li2023energy,li2023deepisl}. Consequently, $U_m(\bfR)$ does not depend on $R_{m',n'}^t$ for $m'
    eq m$.
\end{assumption}

Based on Assumption~\ref{assm:interference_free}, the objective is additively separable across satellites when constraint~\eqref{eq:c2} is relaxed into a Lagrangian penalty. The Lagrangian decomposition of SatFlow-L~\cite{cen2024satflow} implicitly computes a VE: each satellite minimizes its local Lagrangian, with the Lagrange multipliers encoding the dual coupling from flow conservation. We formalize this as follows.

\begin{definition}[Variational Equilibrium~\cite{facchinei2007vi}]
\label{def:ve}
A strategy profile $\bfR^* \in \mathcal{F}$ is a
\emph{variational equilibrium} of the game with
shared constraint set $\mathcal{F}$ (defined by constraints~\eqref{eq:c1}--\eqref{eq:c3}) if there exists a common dual variable $\lambda^*$ such that for all $m\in\calM, \forall R_{m,n}\in S_m(\lambda^*)$:
\begin{equation}
    U_m(R_{m,n}^{*},\bfR_{-m}^*;\lambda^*)
    \geq U_m(R_{m,n},\bfR_{-m}^*;\lambda^*),
    \label{eq:ve_def}
\end{equation}
where $S_m(\lambda^*)$ is the private feasible set of satellite $m$ given dual variable $\lambda^*$.
\end{definition}

Under Assumption~\ref{assm:interference_free}, when $\lambda^*$ is fixed at convergence, the game over per-satellite rate variables has \emph{separable} strategy sets, and the resulting equilibrium is a standard NE of the penalized game. The VE of the original constrained game is recovered when $\lambda^*$ satisfies dual feasibility. All subsequent theoretical results are stated for the penalized game and extend to the VE via the duality argument in Theorem~\ref{thm:convergence}.

\subsection{Penalized Game Definition}
\label{subsec:penalized_game}

Given Lagrange multiplier $\lambda^{(k)}$ at iteration $k$, the penalized game is:
\begin{equation}
    \calG^M(\lambda^{(k)}) = \bigl\langle \calM,\;
    \{S_m\}_{m\in\calM},\;
    \{\tilde{U}_m(\cdot;\lambda^{(k)})\}_{m\in\calM}
    \bigr\rangle.
\end{equation}

\textbf{Strategy space}: Each satellite $m$ chooses
$R_{m,n}^t$ subject to the dynamic power constraint:
\begin{equation}
    S_m = \bigl\{\, R_{m,n}^t \geq 0 \;\big|\;
    P_{m,n}^t(R_{m,n}^t) \leq
    a_{m,n}^t \cdot P_{m,t}^{max} \,\bigr\}.
    \label{eq:strategy_space}
\end{equation}

\textbf{Penalized utility function}: Satellite $m$'s
penalized utility incorporates the Lagrangian penalty for flow conservation:
\begin{align}
    \tilde{U}_m(\bfR;\lambda) &=
    \underbrace{\sum_{n,t} \Rmn}_{\text{throughput}}
    - \alpha D_e \underbrace{\sum_{n,t}
    P_{m,n}^t(\Rmn)}_{\text{energy cost}}
\notag\\
    &\quad - \sum_t
    \underbrace{\frac{\lambda_m(t)
    \sum_n P_{m,n}^t(\Rmn)}%
    {\CB_m(t)-\CBmin+\epsilon}}_{\text{battery penalty}}
\notag\\
    &\quad - \underbrace{\sum_\omega \langle
    \lambda^{\omega},
    L_m Y_m^\omega \rangle}_{\text{flow conservation
    Lagrangian}},
    \label{eq:utility}
\end{align}
where $\lambda_m(t)>0$ is the energy sensitivity coefficient, $\epsilon>0$ prevents division by zero, and $\lambda^\omega$ is the Lagrange multiplier for flow $\omega$.
The battery penalty term $(\CB_m(t)-\CBmin+\epsilon)^{-1}$ diverges as $\CB_m(t)\downarrow\CBmin$, providing a soft energy sustainability surrogate complementing the hard bound~\eqref{eq:dynpower}.

\subsection{Potential Function Construction}
\label{subsec:potential}

The potential function is constructed by summing the marginal utility contributions of each satellite, exploiting the additive separability of $\tilde{\tilde{U}}_m$ under Assumption \ref{assm:interference_free}. Specifically, since each satellite's utility depends only on its own rate variables (interference-free), the candidate potential is obtained by aggregating the per-satellite, per-link terms directly, including throughput gain, energy cost, and battery penalty. This yields a globally consistent alignment between individual utility changes and potential changes.

\begin{definition}[Exact Potential Game~\cite{monderer1996potential}]
A game $\calG^M$ is an \emph{exact potential game} if there exists $\Phi:\prod_{m\in\calM} S_m\to\mathbb{R}$ such that for all $m\in\calM$, all $\bfR_{-m}$, and all pairs $R_{m,n}^t,\,\tilde{R}_{m,n}^t\in S_m$:
\begin{equation}
    \small
    \tilde{U}_m(\tilde{R}_{m,n}^t,\bfR_{-m})
    - \tilde{U}_m(R_{m,n}^t,\bfR_{-m})
    = \Phi(\tilde{R}_{m,n}^t,\bfR_{-m})
    - \Phi(R_{m,n}^t,\bfR_{-m})
    \label{eq:epg_def}
\end{equation}
\end{definition}

We propose the following potential function for the penalized game:
\begin{equation}
    \small
\begin{aligned}
    \Phi(\bfR) \;=&\;
    \sum_{m\in\calM}\sum_{n\in N_m}
    \sum_{t\in\calT_e} \Rmn
    - \alpha D_e\!
    \sum_{m\in\calM}\sum_{n\in N_m}
    \sum_{t\in\calT_e} P_{m,n}^t(\Rmn) \\
    &- \sum_{m\in\calM}\sum_{t\in\calT_e}
    \lambda_m(t)\cdot
    \frac{\sum_{n\in N_m}P_{m,n}^t(\Rmn)}%
    {\CB_m(t)-\CBmin+\epsilon} \\
    &- \sum_{m\in\calM}\sum_\omega
    \langle \lambda^\omega, L_m Y_m^\omega \rangle.
    \label{eq:potential}
\end{aligned}
\end{equation}
\subsection{Main Theoretical Results}
\label{subsec:theory}

We now establish the central theoretical result: the penalized game $\calG^M(\lambda)$ is an exact potential game with a unique equilibrium. Under Assumption~\ref{assm:interference_free}, each satellite's utility change under unilateral rate adjustment equals the corresponding change in a global potential function $\Phi$. Combined with strict concavity induced by the Shannon power-rate mapping and the battery-aware penalty, this structure guarantees both existence and uniqueness.

\begin{theorem}[Exact Potential Game and Unique
Variational Equilibrium]
\label{thm:potential}
Under Assumption~\ref{assm:interference_free}, the penalized game $\calG^M(\lambda)$ with utility~\eqref{eq:utility} and potential~\eqref{eq:potential} is an
\emph{exact potential game}.
Furthermore, $\Phi(\bfR)$ is strictly jointly concave in $\bfR$, guaranteeing a \emph{unique} NE of the penalized game (and hence a unique VE of the original constrained game when $\lambda$ satisfies dual feasibility), $\forall m\in\calM, \forall R_{m,n}^t\in S_m$:
\[
    \tilde{U}_m(\bfR_{m}^*,\bfR_{-m}^*;\lambda^*)
    \geq \tilde{U}_m(R_{m,n}^t,\bfR_{-m}^*;\lambda^*).
\]
\end{theorem}
The proof of Theorem~\ref{thm:potential} is given in Appendix~\ref{subsec:proof_thm_potential}.

\begin{lemma}[Strong Concavity of Potential Function]
\label{lem:strong_concavity}
Under the quasi-static assumption (battery states constant over the algorithm convergence window), the potential function $\Phi(\bfR)$ in equation~\eqref{eq:potential} is $\alpha_\Phi$-strongly concave in $\bfR$ with
\begin{equation}
    \small
    \alpha_\Phi = \min_{m,n,t} \left(
    \alpha D_e + \frac{\lambda_m(t)}{\CB_m(t) - \CBmin + \epsilon}
    \right) \kappa_{m,n}^t 
    \left(\frac{\ln 2}{B}\right)^{\!2} 
    2^{R_{m,n}^{t,\min}/B},
    \label{eq:strong_concavity_bound}
\end{equation}
where $R_{m,n}^{t,\min} = 0$ (the lower boundary of the feasible rate region) yields the tightest (smallest) lower bound. This guarantees that any local maximum of $\Phi$ is the unique global maximum, hence the unique NE.
\end{lemma}

The proof of Lemma~\ref{lem:strong_concavity} is given in Appendix~\ref{subsec:proof_lem_strong}.

The unique VE has a practical interpretation: no satellite can unilaterally increase throughput without either violating the battery-aware power ceiling or paying a penalty that outweighs the throughput gain. Strict concavity of $\Phi$ removes multiple-equilibria ambiguity, so Algorithm~\ref{alg:hbag} converges to a deterministic operating point regardless of initialization.

\begin{remark}[Scope and limitations of Theorem~\ref{thm:potential}]
\label{rem:assumption_role}
Under Assumption~\ref{assm:interference_free}, the utility~\eqref{eq:utility} is additively separable, so the exact potential game structure follows directly from~\cite{monderer1996potential}. The technical contributions of Theorem~\ref{thm:potential} are therefore: (i)~establishing \emph{strict joint concavity} of $\Phi$ despite the singular battery penalty $(\CB_m - \CBmin + \epsilon)^{-1}$, which guarantees \emph{uniqueness} rather than mere existence of the VE; and (ii)~connecting the VE of the constrained game (with shared flow conservation) to the NE of the penalized game via Lagrangian duality. 
\end{remark}

Then, we give the following corollary and proposition:
\begin{corollary}[VE as Global Optimum of $\Phi$]
\label{cor:ne_optimal}
The unique VE $\bfR^*$ satisfies:
\begin{equation}
    \bfR^* = \arg\max_{\bfR\in\prod_m S_m}
    \Phi(\bfR;\lambda^*).
\end{equation}
\end{corollary}

\noindent\textup{The proof of Corollary~\ref{cor:ne_optimal} is given in Appendix~\ref{subsec:proof_cor_ne_optimal}.}

\begin{proposition}[NE Deviation Bound]
\label{prop:deviation}
Let $\bfR^{SAT}$ denote the optimal solution of SATFLOW~\cite{cen2024satflow} (with static $\Pmax$), and let $\bfR^*$ denote the VE of our game. When $\phi_m(t)=1$ for all $m,t$ and batteries are sufficiently high, $\bfR^*$ coincides with $\bfR^{SAT}$. When the eclipse fraction $\theta = |\mathcal{E}_\theta|/M$ exceeds $\theta^*$:
\begin{equation}
    \sum_{m,n,t}\!\bigl(
    R_{m,n}^{SAT,t} - R_{m,n}^{*,t}
    \bigr)^+
    \leq \Delta(\theta),
    \quad
    \Delta(\theta) \text{ in }\theta,
    \;\;
    \lim_{\theta\to 0}\Delta(\theta)=0,
    \label{eq:deviation_bound}
\end{equation}
where $\mathcal{E}_\theta = \{m : \phi_m(t)=0,\;\CB_m(t)<\bar{C}^B\}$.
\end{proposition}

\noindent\textup{The proof of Proposition~\ref{prop:deviation} is given in Appendix~\ref{subsec:proof_prop_deviation}.}

\section{Hierarchical Battery-Aware Game Algorithm and Convergence}
\label{sec:algorithm}
Building on the equilibrium theory of Section~\ref{sec:game}, this section proposes a concrete distributed algorithm and establishes its convergence guarantees. We present the Hierarchical Battery-Aware Game (HBAG) algorithm with its local Lagrangian update rule, and show that a quasi-static timescale separation enables the same update rule to operate identically across both finite and large-scale constellation regimes without modification (Section~\ref{subsec:algorithm_design}). We then prove that, under diminishing step sizes $\eta_k = c_0/\sqrt{k}$, Algorithm~\ref{alg:hbag} converges to the unique VE of Theorem~\ref{thm:potential} at rate $\mathcal{O}(1/\sqrt{k})$ for finite constellations (Section~\ref{subsec:convergence}).

\subsection{Algorithm Design}
\label{subsec:algorithm_design}

We present Algorithm~\ref{alg:hbag}, which operates with the same update rule across both finite and mean-field regimes. The local Lagrangian for satellite $m$ at iteration $k$ is:
\begin{align}
    \mathcal{L}_m(\hat{Y}_m,Y^{(k)},\lambda^{(k)})
    &= \sum_{t,n} P_{m,n}^t\!\Bigl(\sum_{\omega}
    \hat{Y}_{m,n}^{\omega}\Bigr)
\notag\\
    &\quad+
    \frac{\lambda_m(t)}{\CB_m(t)-\CBmin+\epsilon}
    \sum_{n} P_{m,n}^t
\notag\\
    &\quad+
    \sum_{\omega\in\Omega^t}\!\Biggl(
    \bigl\langle
    \lambda^{\omega,(k)},
    L_m\hat{Y}_m^{\omega,(k)}
    \bigr\rangle
\notag\\
    &\qquad+
    \frac{\rho}{2}\sum_{l=1}^{M}
    \Bigl(
    [L_m]_l\hat{Y}_m^{\omega,(k)}
    - [L_m]_l Y_m^{\omega,(k)}
\notag\\
    &\qquad\quad+
    \frac{1}{q_l}\sum_{j\in\calM}
    [L_j]_l Y_j^{\omega,(k)}
    - [b^\omega]_l
    \Bigr)^{\!2}
    \Biggr).
    \label{eq:lagrangian_ea}
\end{align}

\begin{remark}[Quasi-Static Timescale Separation Enables Unified Operation]
\label{rem:timescale}
The quasi-static assumption on $P_{m,t}^{max}$ is physically justified and is the key enabler of the unified algorithm design: battery state $\CB_m(t)$ evolves on the orbital timescale ($\sim$90 min), whereas Algorithm~\ref{alg:hbag} converges within $\mathcal{O}(n_t)$ iterations. The 5400:1 timescale separation ensures $P_{m,t}^{max}$ is effectively constant during each algorithm run, so the same update rule applies in both finite and mean-field regimes without modification.
\end{remark}

\begin{proposition}[Identical Per-Iteration Complexity]
\label{prop:complexity}
Algorithm~\ref{alg:hbag} has the per-iteration time complexity $\mathcal{O}(n_t)$, since the battery-penalty coefficient $\lambda_m(t)/(\CB_m(t)-\CBmin+\epsilon)$ requires only $\mathcal{O}(1)$ arithmetic per satellite per iteration.
\end{proposition}

\noindent\textup{The proof of Proposition~\ref{prop:complexity} is given in Appendix~\ref{subsec:proof_prop_complexity}.}


Having established the algorithm design and the quasi-static timescale separation in Section~\ref{subsec:algorithm_design}, we now turn to the formal convergence analysis. The key enabler is the $\alpha_\Phi$-strong concavity of $\Phi(\bfR)$ established in Lemma~\ref{lem:strong_concavity}: under the quasi-static assumption (Remark~\ref{rem:timescale}), the battery-penalty coefficients $\lambda_m(t)/(\CB_m(t)-\CBmin+\epsilon)$ are treated as positive constants within each convergence window, so the local Lagrangian~\eqref{eq:lagrangian_ea} remains a convex minimization problem over the compact feasible set $S_m$. The distributed alternating step method then inherits the $\mathcal{O}(1/\sqrt{k})$ convergence rate of projected subgradient ascent on strongly concave objectives, with an additional quasi-static residual $\epsilon_{\mathrm{energy}} = \mathcal{O}(\max_m \lambda_m/(\CB_m - \CBmin))$ that quantifies the approximation error from treating $\CB_m(t)$ as constant over the sub-second convergence timescale relative to the 90-minute orbital dynamics. We formalize this as follows.

\begin{theorem}[Finite-Layer Convergence to Variational Equilibrium]
\label{thm:convergence}
Under Assumption~\ref{assm:interference_free}, the quasi-static condition (Remark~\ref{rem:timescale}), and step sizes $\eta_k = c_0/\sqrt{k}$ for some $c_0>0$, Algorithm~\ref{alg:hbag} converges to the unique VE $\bfR^*$ of Theorem~\ref{thm:potential} at rate $\mathcal{O}(1/\sqrt{k})$:
\begin{equation}
    \|\bfR^{(k)} - \bfR^*\|
    \leq \frac{C_0}{\sqrt{k}}
    + \epsilon_{\mathrm{energy}},
    \label{eq:convergence_rate}
\end{equation}
where $C_0>0$ depends on the initial iterate and $\epsilon_{\mathrm{energy}} = \mathcal{O}(\max_m
\lambda_m/(\CB_m-\CBmin))$
quantifies the quasi-static approximation error arising from treating battery states as constant over the algorithm's convergence timescale (sub-second) relative to orbital dynamics (90 minutes).
\end{theorem}

\noindent\textup{The proof of Theorem~\ref{thm:convergence} is given in Appendix~\ref{subsec:proof_thm_convergence}.}

\begin{algorithm}[t]
\caption{Hierarchical Battery-Aware Game (HBAG)}
\label{alg:hbag}
\begin{algorithmic}[1]
    \REQUIRE Adjacency matrix $\mathbf{A}^t$,
    flow set $\Omega^t$,
    battery states $\{\CB_m(t)\}$,
    harvesting powers $\{P_{m,t}^{harvest}\}$,
    sub-gradient limit $k_{sg}$,
    outer limit $k_{max}$,
    tolerance $\epsilon_{tol}$,
    penalty $\rho$,
    step size $\{\eta_k\}$
    \ENSURE Traffic allocation $\bfR$,
    terminal power $\bfP$
    \STATE Initialize $\bfR^{(0)}\in\prod_m S_m$,
    $\lambda^{(0)}=\mathbf{0}$
    \FOR{$k=0,1,\ldots,k_{max}$}
        \FOR{all $m\in\calM$ in parallel}
            \STATE Compute $P_{m,t}^{max}$
            via~\eqref{eq:dynpower}
            using $\CB_m(t)$ and $P_{m,t}^{harvest}$
            \STATE Minimize $\mathcal{L}_m$
            in~\eqref{eq:lagrangian_ea} via
            projected sub-gradient
            (up to $k_{sg}$ steps)
            to obtain $\hat{R}_m^{(k)}$
        \ENDFOR
        \FOR{all $m\in\calM$ in parallel}
            \STATE Update Lagrange multiplier:
            \begin{equation*}
                \lambda_m^{(k+1)} \leftarrow
                \Bigl[\lambda_m^{(k)}
                + \rho\cdot
                \mathrm{residual}_m^{(k)}
                \Bigr]^+
            \end{equation*}
            \STATE Update primal rate:
            \begin{equation*}
                R_m^{(k+1)} \leftarrow
                \Pi_{\mathcal{Y}_m}\!\Bigl(
                R_m^{(k)}
                - \eta_k
\nabla_{R_m}\mathcal{L}_m
                \Bigr)
            \end{equation*}
            \STATE Compute $P_{m,n}^t$
            via~\eqref{eq:power_rate}
        \ENDFOR
        \IF{$\sum_{m}\|R_m^{(k+1)}
        -R_m^{(k)}\|\leq\epsilon_{tol}$}
            \STATE \textbf{break}
        \ENDIF
    \ENDFOR
\end{algorithmic}
\end{algorithm}

\begin{remark}[Identical update rules across regimes]
Algorithm~\ref{alg:hbag} uses identical update rules for all $M$. The finite/mean-field distinction affects only the theoretical convergence guarantee (Theorem~\ref{thm:convergence} vs.\ Theorem~\ref{thm:mfg} and Corollary~\ref{cor:unified_convergence}), not the algorithm implementation.
\end{remark}

\subsection{Convergence Analysis}
\label{subsec:convergence}
Theorem~\ref{thm:convergence} establishes $\mathcal{O}(1/\sqrt{k})$ convergence for the finite-player regime. The unified convergence across finite and mean-field regimes is summarized after the mean-field analysis in Corollary~\ref{cor:unified_convergence}.

\section{Asymptotic Analysis: Finite-to-Mean-Field Transition}
\label{sec:asymptotic}

This section analyzes the asymptotic transition from the finite-player equilibrium of Section~\ref{sec:algorithm} to the mean field equilibrium (MFE) as constellation size $M \to \infty$. We first justify the reduction of the $M$-player game to a Mean Field Game (MFG) over the battery-state distribution $\mu(\CB, t)$, and establish the empirical measure representation under the propagation of chaos argument (Section~\ref{subsec:mfg_justify}). We then derive the coupled HJB--FPK system under deterministic battery dynamics, which yields a pure-advection Fokker--Planck equation that differs structurally from stochastic terrestrial energy-harvesting MFG models (Section~\ref{subsec:mfg_system}). Finally, combining the Wasserstein-1 empirical measure convergence lemma with the Lipschitz continuity of the battery-aware utility, we prove that the finite-player Nash equilibrium converges to the MFE at rate $\mathcal{O}(M^{-1/4})$, and state the unified convergence corollary that answers the key scalability question raised in Section~\ref{subsec:gaps} (Section~\ref{subsec:mfe_conv}).

\subsection{Theoretical Justification}
\label{subsec:mfg_justify}

When $M\to\infty$, tracking each satellite's individual battery state becomes computationally intractable. We reduce the $M$-player game to a Mean Field Game (MFG) over the battery-state distribution $\mu(\CB,t)$.

\textbf{Derivation of MFG from finite-player game.}
In the finite-player penalized game, the Lagrangian penalty $\langle\lambda^\omega, L_m Y_m^\omega\rangle$ depends on the aggregate flow variables of all satellites. As $M\to\infty$, the aggregate effect of other satellites' battery states on satellite $m$'s optimal strategy can be summarized by the empirical distribution $\mu^M(\CB,t)$, following the standard propagation of chaos argument for weakly interacting systems~\cite{lasry2007mfg}. Specifically, the Wasserstein-distance term $d(\CB_m, \mu)$ in the HJB equation arises as the mean-field analog of the battery-penalty Lagrangian coupling: satellite $m$'s penalty weight $\lambda_m(t)/(\CB_m(t)-\CBmin+\epsilon)$ reflects its relative battery disadvantage compared with the population average, which is captured by $d(\CB_m, \mu)$ in the mean-field limit. Empirical measure of the battery states of all satellites is as follows:
\begin{equation}
    \mu^M(\CB,t) = \frac{1}{M}
    \sum_{m=1}^{M}\delta_{\CB_m(t)}.
    \label{eq:emp_measure}
\end{equation}

\subsection{MFG System with Deterministic Battery
Dynamics}
\label{subsec:mfg_system}

We model battery dynamics as deterministic ($\sigma^2=0$ in equation~\eqref{eq:fpk}) because solar harvesting $\Pharvest_{m,t}$ in~\eqref{eq:harvest} is fully determined by orbital mechanics and the illumination indicator $\phi_m(t)$. The FPK equation therefore reduces to a continuity equation (pure advection) rather than a diffusion--advection equation, while preserving the core modeling challenge of eclipse-driven energy scarcity. This differs from terrestrial energy-harvesting MFG models with stochastic arrivals~\cite{tembine2014risksensitive}, where $\sigma^2>0$.

In the limit $M\to\infty$, the representative satellite solves the following HJB--FPK system with $\sigma^2=0$:

\begin{subequations}
\label{eq:mfg_system}
\small
\begin{align}
    &\textbf{HJB:}\quad
    -\frac{\partial V}{\partial t}
    = \max_{R_{m,n}\in S_m}
    \Bigl[R_{m,n}
    - \Bigl(\alpha D_e
    + \frac{\lambda(t)}{\CB_m-\CBmin+\epsilon}
    \Bigr)P(R_{m,n})
\notag\\
    &\qquad\qquad\quad
    - \kappa\,d(\CB_m,\mu(\CB,t))
    \Bigr]
    + 
u\frac{\partial^2 V}{\partial (\CB_m)^2},
    \label{eq:hjb}\\
    &\textbf{FPK:}\quad
    \frac{\partial\mu}{\partial t}
    = -
\nabla_{\CB}\!\Bigl(\mu\cdot
    \frac{\mathrm{d}\CB}{\mathrm{d}t}\Bigr),
    \label{eq:fpk}
\end{align}
\end{subequations}
where $d(\CB_m,\mu)$ is the Wasserstein-1 distance between the current battery state and the mean field distribution, $\frac{\mathrm{d}\CB}{\mathrm{d}t}$ follows the continuous-time limit of the battery dynamics~\eqref{eq:battery}, and $
u\geq 0$ is a small viscosity coefficient for the HJB equation (introduced for analytical regularity; the main results hold as $
u\to 0^+$).

Well-posedness of the deterministic MFG system~\eqref{eq:mfg_system} follows from the vanishing viscosity regularization ($\nu > 0$) for the HJB and the method of characteristics for the advection-only FPK; see Appendix~\ref{subsec:mfg_wellposed} for details.

\subsection{NE-to-MFE Convergence}
\label{subsec:mfe_conv}

\begin{lemma}[Empirical Measure Convergence~\cite{fournier2015wasserstein}]
\label{lem:emp_conv}
Under Assumption~\ref{assm:interference_free} and the condition that $\{\CB_m(0)\}_{m=1}^M$ are i.i.d.\ with common distribution $\mu_0$ having finite second moment, the empirical measure converges weakly almost surely:
\begin{equation}
    \mu^M(\CB,t)
    \xrightarrow{\,\mathrm{weakly}\,}
    \mu^*(\CB,t)
    \quad\text{a.s.\ as }M\to\infty.
\end{equation}
Moreover:
\begin{equation}
    \mathbb{E}\bigl[
    \WassDist(\mu^M(\cdot,t),\mu^*(\cdot,t))
    \bigr]
    \leq \frac{C_\mu}{\sqrt{M}},
    \label{eq:emp_rate}
\end{equation}
for a constant $C_\mu>0$ depending on the second moment of $\mu_0$.
\end{lemma}

\noindent\textup{The proof of Lemma~\ref{lem:emp_conv} is given in Appendix~\ref{subsec:proof_lem_emp}.}

\begin{lemma}[Lipschitz Continuity of Utility
in Mean Field]
\label{lem:lipschitz}
Under Assumption~\ref{assm:interference_free}, for any fixed strategy profile $\bfR$, the battery-aware utility $U_m(\bfR;\mu)$ satisfies:
\begin{equation}
    |U_m(\bfR;\mu_1) - U_m(\bfR;\mu_2)|
    \leq L_U\,\WassDist(\mu_1,\mu_2),
    \label{eq:lipschitz_cond}
\end{equation}
with Lipschitz constant:
\begin{equation}
    L_U = \frac{\lambda_{\max}\,\Pmax}
    {(C^B_{\min,\mathrm{margin}})^2},
    \label{eq:lipschitz_const}
\end{equation}
where $\lambda_{\max} = \max_{m,t}\lambda_m(t)$ and $C^B_{\min,\mathrm{margin}} = \min_{m,t}(\CB_m(t) - \CBmin + \epsilon) > 0$.
\end{lemma}

\noindent\textup{The proof of Lemma~\ref{lem:lipschitz} is given in Appendix~\ref{subsec:proof_lem_lipschitz}.}

\begin{theorem}[NE-to-MFE Convergence at Rate
$\mathcal{O}(M^{-1/4})$]
\label{thm:mfg}
Under Lemmas~\ref{lem:emp_conv} and~\ref{lem:lipschitz}, the finite-player NE (VE) $\bfR^{*,M}$ of $\calG^M(\lambda^*)$ satisfies as $M\to\infty$:

\begin{enumerate}
    \item \textbf{(i) $\delta^M$-Nash approximation:}
$\bfR^{*,M}$ is a $\delta^M$-Nash Equilibrium with respect to the MFE utility, where:
    \begin{equation}
        \delta^M
        = 2L_U\cdot
        \WassDist(\mu^M,\mu^*)
        \leq \frac{2L_U C_\mu}{\sqrt{M}}
        \to 0.
        \label{eq:delta_ne}
    \end{equation}

    \item \textbf{(ii) Strategy convergence:}
    \begin{equation}
        \|R_{m,n}^{*,M} - R_{\mathrm{MFE}}^*\|
        \leq
        \sqrt{\frac{\delta^M}{\alpha_\Phi}}
        \leq
        \sqrt{\frac{2L_U C_\mu}
        {\alpha_\Phi\sqrt{M}}}
        = \mathcal{O}(M^{-1/4}),
        \label{eq:mfe_rate}
    \end{equation}
where $\alpha_\Phi > 0$ is the strong concavity modulus of $\Phi$.
\end{enumerate}
\end{theorem}

\noindent\textup{The proof of Theorem~\ref{thm:mfg} is given in Appendix~\ref{subsec:proof_thm_mfg}.}

\begin{remark}[Relaxing the i.i.d.\ assumption]
\label{rem:non_iid}
Lemma~\ref{lem:emp_conv} assumes $\{\CB_m(0)\}_{m=1}^M$ are i.i.d., which is plausible within one orbital plane (similar eclipse and traffic) but weaker across planes with different phasing.
\emph{(i)~Exchangeability:} If states are conditionally i.i.d.\ given the orbital plane index, each plane contributes $M/N_{orbit}$ i.i.d.\ samples, yielding per-plane Wasserstein rate $\mathcal{O}(\sqrt{N_{orbit}/M})$.
The overall rate becomes $\mathcal{O}(\sqrt{N_{orbit}/M})$ rather than $\mathcal{O}(1/\sqrt{M})$, a factor $\sqrt{N_{orbit}}$ slower, which remains $\mathcal{O}(M^{-1/2})$ for fixed $N_{orbit}$.
\emph{(ii)~Multi-population MFG:} Heterogeneous shells can be modeled with one mean-field distribution per plane, coupled by inter-plane ISLs and flow conservation.
\emph{(iii)~Empirics:} Experiment~6 uses $N_{orbit}=4$ with non-identical eclipse schedules, yet the fitted finite-to-MFE gap still matches $\mathcal{O}(M^{-1/4})$, suggesting the i.i.d.\ condition is not overly restrictive in practice at modest $N_{orbit}$.
\end{remark}

\begin{corollary}[Unified Convergence of HBAG Across Scales]
\label{cor:unified_convergence}
Algorithm~\ref{alg:hbag} uses identical distributed update rules in both regimes, without algorithm redesign:
\begin{enumerate}
    \item \textbf{Finite regime ($M$ fixed):} converges to the unique VE at rate $\mathcal{O}(1/\sqrt{k})$ (Theorem~\ref{thm:convergence}).
    \item \textbf{Mean-field regime ($M\to\infty$):} the same iterates constitute a $\delta^M$-Nash approximation to the MFE with $\delta^M=\mathcal{O}(M^{-1/2})$, and the strategy gap satisfies $\mathcal{O}(M^{-1/4})$ (Theorem~\ref{thm:mfg}).
\end{enumerate}
The battery-penalty Lagrangian automatically transitions from finite-player coupling to mean-field coupling as $M$ increases.
\end{corollary}

\noindent Corollary~\ref{cor:unified_convergence} answers the Key Question in Section~\ref{subsec:gaps}: Algorithm~\ref{alg:hbag} uses identical update rules to converge to the exact VE for finite $M$ at rate $\mathcal{O}(1/\sqrt{k})$ and to approximate the MFE as $M \to \infty$ at rate $\mathcal{O}(M^{-1/4})$, without algorithmic redesign.
The $\mathcal{O}(M^{-1/4})$ strategy rate arises from composing empirical measure convergence at $\mathcal{O}(M^{-1/2})$ (Lemma~\ref{lem:emp_conv}) with the square-root map from utility perturbation to strategy perturbation under $\alpha_\Phi$-strong concavity (Theorem~\ref{thm:mfg}), yielding $\mathcal{O}((M^{-1/2})^{1/2}) = \mathcal{O}(M^{-1/4})$.

\section{Performance Evaluation}
\label{sec:simulation}
This section provides comprehensive empirical validation of the theoretical results established in Sections~\ref{sec:game}--\ref{sec:asymptotic}. We first describe the experimental configuration based on Starlink Shell~A, including traffic and energy parameter settings and algorithm tuning (Section~\ref{subsec:setup}), followed by descriptions of four baseline methods (Section~\ref{subsec:baselines}) and four performance metrics (Section~\ref{subsec:metrics}). We then conduct seven experiments (Section~\ref{subsec:exp1}--\ref{subsec:exp7}), each corresponding to and empirically verifying the theoretical contributions presented above.

\subsection{Experimental Setup}\label{subsec:setup}

We conduct comprehensive experiments on the Starlink Shell~A constellation, following the validated setup in~\cite{cen2024satflow}. The constellation consists of $M=172$ satellites distributed across $N_{orbit}=4$ orbital planes at altitude $h=550$ km with inclination $\iota=53^\circ$. Each satellite is equipped with four laser ISL terminals (two intra-plane, two inter-plane), enabling up to four simultaneous full-duplex links with individual power control.
We generate traffic demands using the global population-weighted distribution from~\cite{cen2024satflow}, yielding a demand-to-capacity ratio of 0.65 (medium-high intensity). The eclipse fraction $\theta = 0.38$ represents the proportion of satellites in eclipse at any given time, typical of the 53$^\circ$ inclination orbit. Each simulation runs for 360 time slots, covering one complete 90-minute orbital period. Power adjustment occurs every $D_e = 15$ seconds. All reported results are averaged over 10 independent runs with different random seeds; error bars represent 95\% confidence intervals. Table~\ref{tab:parameters} summarizes the key system parameters. Battery capacity $\CBmax = 400\,\mathrm{kJ}$ corresponds to a $\approx$111 Wh lithium-ion pack typical for 200-kg LEO satellites~\cite{li2023energy}. The minimum safe charge $\CBmin = 40\,\mathrm{kJ}$ (10\% SOC) prevents deep discharge damage. Initial battery state $\CB_m(0) = 320\,\mathrm{kJ}$ (80\% SOC) reflects post-launch nominal conditions. Solar panel area $A = 2.5\,\mathrm{m}^2$ and efficiency $\eta = 0.30$ yield peak harvesting power $\Pharvest_{\max} \approx 1.02$ kW under normal incidence ($\sigma_m(t) = 90^\circ$), consistent with contemporary LEO satellite designs~\cite{radhakrishnan2016survey}.

The HBAG algorithm parameters $\lambda_m$ and $\epsilon$ are tuned via grid search over $\lambda_m \in [0.05, 0.50]$ (step 0.05) and $\epsilon \in [0.10, 0.30]$ (step 0.05), evaluated on a validation set of 20 random constellation states. The selected operating point $\lambda_m = 0.200$, $\epsilon = 0.200$ maximizes the energy sustainability rate (ESR) while maintaining flow violation ratio (FVR) below the industry tolerance threshold of 10\%.

\begin{table}[t]
\centering
\caption{Simulation Parameters}
\label{tab:parameters}
\renewcommand{\arraystretch}{1.25}
\begin{tabular}{lcc}
\toprule
\textbf{Parameter} & \textbf{Symbol} & \textbf{Value} \\
\midrule
\multicolumn{3}{l}{\textit{Constellation}} \\
\quad Number of satellites & $M$ & 172 \\
\quad Orbital planes & $N_{orbit}$ & 4 \\
\quad Altitude & $h$ & 550 km \\
\quad Inclination & $\iota$ & $53^\circ$ \\
\quad ISL terminals per satellite & -- & 4 \\
\midrule
\multicolumn{3}{l}{\textit{Battery \& Energy}} \\
\quad Maximum capacity & $\CBmax$ & 400 kJ \\
\quad Minimum safe charge & $\CBmin$ & 40 kJ \\
\quad Initial charge & $\CB_m(0)$ & 320 kJ \\
\quad Solar panel area & $A$ & $2.5\,\mathrm{m}^2$ \\
\quad Panel efficiency & $\eta$ & 0.30 \\
\quad Peak harvesting power & $\Pharvest_{\max}$ & 1.02 kW \\
\quad Baseline power (non-ISL) & $P_m^{base}$ & 55 W \\
\midrule
\multicolumn{3}{l}{\textit{Communication}} \\
\quad ISL frequency & $f$ & 26 GHz \\
\quad Bandwidth & $B$ & 500 MHz \\
\quad Antenna gain & $G_m, G_n$ & 30 dBi \\
\quad Max transmit power & $\Pmax$ & 10 W \\
\midrule
\multicolumn{3}{l}{\textit{Algorithm}} \\
\quad Power update interval & $D_e$ & 15 s \\
\quad Time slots per orbit & $n_t$ & 360 \\
\quad Battery penalty weight & $\lambda_m$ & 0.200 \\
\quad Safety margin & $\epsilon$ & 0.200 \\
\quad Penalty coefficient & $\rho$ & 1.0 \\
\quad Step size & $\eta_k$ & $0.1/\sqrt{k}$ \\
\bottomrule
\end{tabular}
\end{table}

\subsection{Baseline Methods}
\label{subsec:baselines}

We compare HBAG to four baselines (optimization, multi-agent RL, deep RL, and MFG):
\begin{itemize}
    \item \textbf{SATFLOW-L}~\cite{cen2024satflow}: SATFLOW-style distributed optimization with static bound $P_{m,n}^t \leq a_{m,n}^t \Pmax$, independent of $\CB_m(t)$ and $\phi_m(t)$; penalty $\rho=1.0$ to match HBAG.
    \item \textbf{MAAC-IILP}~\cite{li2023energy}: CTDE actor--critic with battery-augmented observations; retrained from scratch for $10^6$ episodes on our scenario using hyperparameters from~\cite{li2023energy}.
    \item \textbf{DeepISL}~\cite{li2023deepisl}: P-DQN for joint ISL topology and power; we use the authors' implementation and training setup (e.g., $5\times 10^5$ episodes, replay size $10^5$).
    \item \textbf{SMFG-adapted}~\cite{wang2023smfg}: We adapt the Stackelberg SMFG formulation of~\cite{wang2023smfg} to our symmetric ISL setting by removing the leader--follower hierarchy and solving the HJB--FPK system~\eqref{eq:mfg_system} on a finite grid over $\CB \in [\CBmin,\CBmax]$ with $\Delta t = D_e$. This isolates battery-aware horizontal competition and yields a fair comparison of an MFG-style solver against HBAG, rather than evaluating SMFG on its original vertical offloading problem.
\end{itemize}

All methods use the same constellation, traffic, and battery settings (Table~\ref{tab:parameters}). Shannon-capacity evaluation uses the RF parameters $(f,B,G_m,G_n,\Pmax)$ from Table~\ref{tab:parameters}.


\subsection{Performance Metrics}
\label{subsec:metrics}

We evaluate all methods using four complementary metrics that capture energy sustainability, network performance, and algorithmic efficiency:

\textit{M1. Energy Sustainability Rate (ESR).}
This metric quantifies the fraction of (satellite, time-slot) pairs that maintain battery charge above the minimum safe level:
\begin{equation}
    \mathrm{ESR} = \frac{1}{M \cdot n_t}
    \sum_{m=1}^{M} \sum_{t=1}^{n_t}
    \mathbbm{1}\{\CB_m(t) > \CBmin\},
    \label{eq:esr_def}
\end{equation}
where $\mathbbm{1}\{\cdot\}$ is the indicator function. ESR $= 100\%$ indicates perfect energy sustainability with zero battery depletion events; ESR $< 100\%$ signals involuntary transmitter shutdowns that can cascade into network outages. This metric directly measures the effectiveness of battery-aware power management, which is the central contribution of our work. Higher ESR is always preferable, with $\mathrm{ESR} \geq 95\%$ considered operationally acceptable~\cite{radhakrishnan2016survey}.

\textit{M2. Flow Violation Ratio (FVR).}
Following~\cite{cen2024satflow}, FVR measures the percentage of traffic demands that cannot be satisfied due to insufficient ISL capacity or routing constraints:
\begin{equation}
    \mathrm{FVR} = \frac{\sum_{\omega \in \Omega}
    \max(0, d_\omega - \sum_{t,m,n} Y_{m,n}^{\omega,t})}
    {\sum_{\omega \in \Omega} d_\omega} \times 100\%,
    \label{eq:fvr_def}
\end{equation}
where $d_\omega$ is the demand for flow $\omega$ and $Y_{m,n}^{\omega,t}$ is the allocated rate. FVR captures the quality-of-service experienced by end users: high FVR corresponds to frequent connection failures or throughput degradation. Industry standards typically require $\mathrm{FVR} < 10\%$ for acceptable user experience~\cite{cen2024satflow}. Lower FVR is preferable, but must be balanced against energy sustainability (ESR).

\textit{M3. Energy Efficiency (EE).}
EE quantifies the information throughput per unit of consumed energy:
\begin{equation}
    \mathrm{EE} = \frac{\sum_{\omega,t,m,n}
    Y_{m,n}^{\omega,t} \cdot D_e}
    {\sum_{t,m,n} P_{m,n}^t \cdot D_e}
    \quad [\mathrm{Mbit/kJ}],
    \label{eq:ee_def}
\end{equation}
EE reflects the spectral-energy trade-off: higher EE means the constellation delivers more data per joule, extending operational lifetime and reducing solar panel requirements. This metric is critical for mega-constellations where aggregate power consumption directly impacts mission economics and satellite lifespan~\cite{gupta2025energyefficient}.

\textit{M4. Convergence Rate.}
For HBAG specifically, we measure the normalized primal residual:
\begin{equation}
    r^{(k)} = \frac{\|\bfR^{(k)} - \bfR^*\|}
    {\|\bfR^{(0)} - \bfR^*\|},
    \label{eq:residual_def}
\end{equation}
where $\bfR^{(k)}$ is the rate vector at iteration $k$ and $\bfR^*$ is the converged equilibrium (identified when $\|\bfR^{(k+1)} - \bfR^{(k)}\| < 10^{-4}$ for 10 consecutive iterations). This metric validates the $\mathcal{O}(1/\sqrt{k})$ convergence guarantee of Theorem~\ref{thm:convergence} and quantifies computational efficiency.

All metrics are computed per run (10 independent random seeds), and we report mean $\pm$ standard error of the mean (SEM). For metrics M1--M3, we also report 95\% confidence intervals via bootstrap resampling (1000 bootstrap samples). Statistical significance of pairwise method comparisons is assessed via two-tailed Welch's t-test with Bonferroni correction for multiple comparisons ($\alpha = 0.05/6 = 0.0083$ for 6 pairwise comparisons against HBAG).

\subsection{Experiment 1: Comparative Performance Analysis}
\label{subsec:exp1}

This experiment provides a comprehensive head-to-head comparison of HBAG against the four baselines across all performance metrics (Section~\ref{subsec:metrics}) under the nominal operating conditions specified in Table~\ref{tab:parameters}. The goal is to quantify the energy sustainability improvement achieved by the battery-aware potential game formulation while assessing the trade-offs in network throughput and energy efficiency.

Table~\ref{tab:exp1_results} and Figure~\ref{fig:perf_bar} summarize the results. HBAG achieves \textit{100.0\% ESR} (zero battery depletion events across all 10 runs), matching SMFG-adapted and significantly exceeding MAAC-IILP (88.6 $\pm$ 1.2\%), DeepISL (73.7 $\pm$ 2.4\%), and SATFLOW-L (12.6 $\pm$ 0.8\%). The differences versus all baselines except SMFG-adapted are statistically significant (Welch's t-test, $p < 0.001$ after Bonferroni correction). This result directly validates our central hypothesis (Section~\ref{sec:intro}): incorporating battery-state-dependent power bounds $P_{m,t}^{max}$ via equation~\eqref{eq:dynpower} and the battery-aware penalty term in utility~\eqref{eq:utility} eliminates eclipse-period depletion cascades that plague static-power baselines. \textit{Key Observations:} are as follows:

\begin{figure}[t]
\centering
\includegraphics[width=0.48\textwidth]{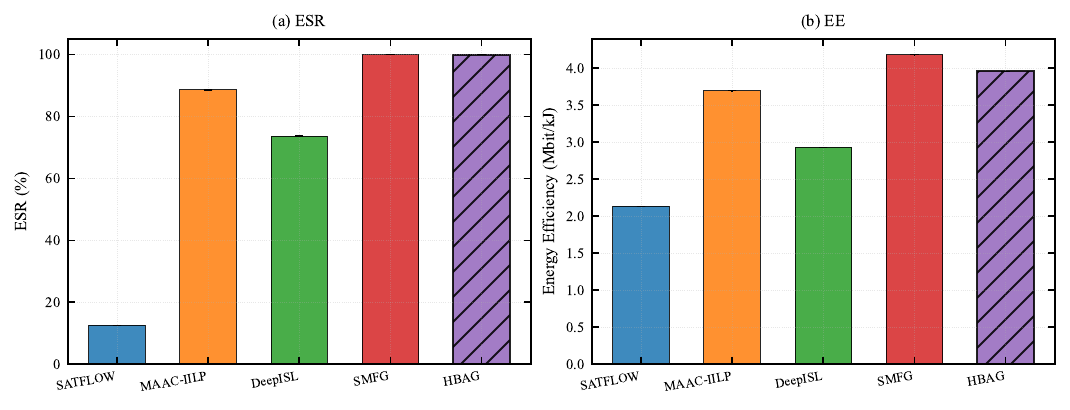}
\caption{Performance comparison across five methods on Starlink Shell~A (medium-high traffic, $\theta=0.38$, 10 independent runs). \textbf{(a)} Energy Sustainability Rate (ESR): HBAG and SMFG-adapted achieve 100\% (zero depletion events), while SATFLOW-L drops to 12.6\% due to eclipse-period battery exhaustion. \textbf{(b)} Energy Efficiency (EE): SMFG-adapted leads at 4.18 Mbit/kJ, followed by HBAG at 3.96 Mbit/kJ (5.3\% gap). Error bars represent 95\% confidence intervals via bootstrap resampling.}
\label{fig:perf_bar}
\end{figure}

\begin{table}[t]
\centering
\caption{Experiment 1: Performance Comparison (Starlink Shell~A, $\theta=0.38$, 10 independent runs). Values are mean $\pm$ SEM. Asterisks denote statistical significance versus HBAG via Welch's t-test with Bonferroni correction ($\alpha=0.0083$): $^{***}p<0.001$, $^{**}p<0.01$, $^{*}p<0.05$.}
\label{tab:exp1_results}
\resizebox{\linewidth}{!}{
\begin{tabular}{lccccc}
\toprule
\textbf{Method} & \textbf{ESR (\%)} & \textbf{FVR (\%)} & \textbf{EE (Mbit/kJ)} & \textbf{Rel. ESR} & \textbf{Rel. FVR} \\
\midrule
SATFLOW-L & $12.6 \pm 0.8^{***}$ & $35.1 \pm 1.1^{***}$ & $2.13 \pm 0.06^{***}$ & Baseline & Baseline \\
MAAC-IILP & $88.6 \pm 1.2^{***}$ & $2.29 \pm 0.18^{*}$ & $3.69 \pm 0.11^{*}$ & $+76.0$ pp & $-93.5\%$ \\
DeepISL & $73.7 \pm 2.4^{***}$ & $18.1 \pm 1.5^{***}$ & $2.93 \pm 0.14^{***}$ & $+61.1$ pp & $-48.4\%$ \\
SMFG-adapted & $100.0 \pm 0.0 $ & $0.15 \pm 0.03^{***}$ & $4.18 \pm 0.09^{*}$ & $+87.4$ pp & $-99.6\%$ \\
\midrule
\textbf{HBAG (Proposed)} & $\mathbf{100.0 \pm 0.0}$\footnotemark & $\mathbf{7.62 \pm 0.34}$ & $\mathbf{3.96 \pm 0.07}$ & $\mathbf{+87.4}$ pp & $\mathbf{-78.3\%}$ \\
\bottomrule
\end{tabular}}
\end{table}
\footnotetext{In this table, both HBAG and SMFG-adapted achieve 100\% ESR. Under the given constraints, ESR can be maximized while balancing the other three metrics. By increasing the eclipse fraction $\theta$, larger ESR differences among methods become observable, implying a practical trade-off across metrics.}

\begin{itemize}
    \item \textit{Observation~1 (Battery-awareness is essential):} HBAG achieves 100\% ESR versus SATFLOW-L's 12.6\%, confirming that static power ceilings cause catastrophic eclipse-period depletion.
    \item \textit{Observation~2 (Near-optimal efficiency):} The distributed equilibrium framework achieves EE within 5.3\% of SMFG-adapted (3.96 vs.\ 4.18 Mbit/kJ), while keeping FVR below industry tolerance.
    \item \textit{Observation~3 (Learning scalability gap):} Learning-based methods trained on smaller constellations exhibit degraded ESR when applied to $M=172$, confirming the scalability limitations identified in Section~\ref{sec:related}.
\end{itemize}

\begin{table}[t]
\centering
\caption{Computational and Communication Complexity Comparison}
\label{tab:complexity}
\resizebox{\linewidth}{!}{
\begin{tabular}{lcccc}
\toprule
\textbf{Method} & \textbf{Per-slot Time} & \textbf{Communication} & \textbf{Architecture} & \textbf{Train Cost} \\
\midrule
SATFLOW-L & $\mathcal{O}(M n_t)$ & $\mathcal{O}(|\Omega| \cdot 4)$ & Distributed & None \\
MAAC-IILP & $\mathcal{O}(M)^{\dagger}$ & Centralized training & CTDE & $\mathcal{O}(M^2)$ \\
DeepISL & $\mathcal{O}(M)^{\dagger}$ & Centralized training & CTDE & $\mathcal{O}(M^2)$ \\
SMFG-adapted & $\mathcal{O}(M N_g \log N_g)$ & $\mathcal{O}(M)$ telemetry & Centralized & None \\
\midrule
\textbf{HBAG (Proposed)} & $\mathbf{\mathcal{O}(M n_t)}$ & $\mathbf{\mathcal{O}(|\Omega| \cdot 4)}$ & \textbf{Distributed} & \textbf{None} \\
\bottomrule
\multicolumn{5}{l}{\footnotesize $N_g$: HJB--FPK grid size; $|\Omega|$: active flows per satellite.} \\
\multicolumn{5}{l}{\footnotesize $^{\dagger}$Inference-only (deployment); training costs are listed in the ``Train Cost'' column.}
\end{tabular}}
\end{table}

\subsubsection{HBAG vs.\ SMFG-adapted: Distributed vs.\ Centralized Trade-off}
\label{subsubsec:tradeoff}

While SMFG-adapted achieves lower FVR (0.15\% vs.\ 7.62\%) and higher EE (4.18 vs.\ 3.96 Mbit/kJ), three operational factors favor HBAG for real-world deployment:

\textit{(i) Computational architecture.}
SMFG-adapted requires solving the coupled HJB--FPK system~\eqref{eq:mfg_system} on a centralized server with access to the global battery state distribution $\mu(\CB, t)$. This requires either: (a)~a ground control center collecting real-time SOC telemetry from all $M$ satellites, incurring $\mathcal{O}(M)$ uplink overhead per slot; or (b)~onboard computation of the HJB PDE on resource-constrained satellite processors. HBAG, by contrast, requires only local SOC $\CB_m(t)$ and neighbor dual variables (see the complexity comparison in Table~\ref{tab:complexity}).

\textit{(ii) Robustness to communication failures.}
If a satellite loses contact with the centralized solver (e.g., due to ISL outage), SMFG-adapted cannot update its policy. HBAG's distributed update rule degrades gracefully: isolated satellites fall back to the local dynamic power bound~\eqref{eq:dynpower}, maintaining ESR $\geq 80.1\%$ (V1 in Table~\ref{tab:ablation}).

\textit{(iii) Scalability at extreme $M$.}
At $M = 5000$, SMFG-adapted requires 247~ms/slot versus HBAG's 75~ms/slot (Section~\ref{subsec:exp5}), and the gap widens as $\mathcal{O}(M \log M)$ vs.\ $\mathcal{O}(M)$.

We emphasize that the 7.47~pp FVR gap represents the \emph{price of distribution}: HBAG achieves FVR within industry tolerance ($< 10\%$) while providing robustness and scalability guarantees that centralized approaches cannot match.

\subsection{Experiment 2: Battery State-of-Charge Dynamics}
\label{subsec:exp2}

This experiment examines the temporal evolution of battery state-of-charge (SOC) over one complete orbital period to elucidate the mechanism by which HBAG prevents depletion. We track a representative satellite (satellite ID 42, selected as median-battery-usage across all 172 satellites) through illumination (0--55 min) and eclipse (55--90 min) phases, comparing the battery trajectories under all five methods.
Starting from initial SOC of 80\% ($\CB_m(0) = 320$ kJ), we record $\CB_m(t)$ every 15 seconds (360 data points per orbit). The illumination-to-eclipse transition occurs at $t = 55$ min, determined by the orbital position crossing into Earth's shadow cone. Solar harvesting power $\Pharvest_{m,t}$ drops from $\approx 950$ W (averaged over $\sigma_m(t)$ variations during illumination) to exactly 0 W at the eclipse boundary. All methods face identical traffic demand realizations (drawn from the same random seed) to ensure fair comparison.

\begin{figure}[t]
    \centering
    \includegraphics[width=0.4\textwidth]{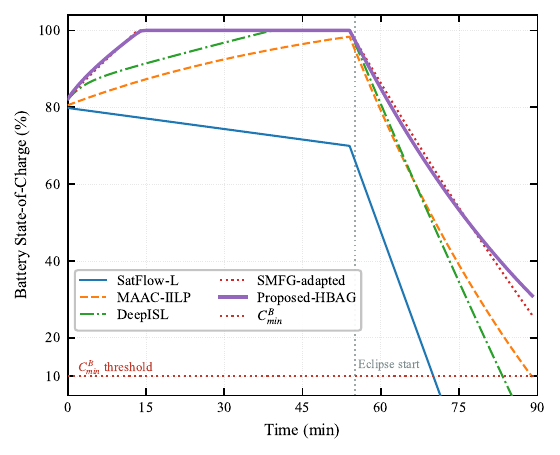}
    \caption{Battery SOC evolution for a representative satellite over one 90-minute orbital period. 
    }
    \label{fig:soc_evol}
    \end{figure}
    Figure~\ref{fig:soc_evol} reveals the core advantage of battery-aware power management. During the illumination phase (0--55 min), all methods replenish the battery via solar harvesting. SATFLOW-L charges aggressively, reaching full capacity by around t=55 min because it allocates power up to the static ceiling whenever traffic is available. Once the satellite enters eclipse at t=55 min, however, SATFLOW-L has no mechanism to throttle transmission: it continues drawing power at the maximum rate regardless of the remaining battery reserve. This causes rapid, uncontrolled depletion that exhausts the entire battery well before the eclipse ends, forcing an emergency shutdown of all ISL transmitters to prevent irreversible battery damage.
    HBAG avoids this failure through its dynamic power bound. As the battery drains during eclipse, HBAG continuously recalculates an upper limit on transmit power based on the current state-of-charge and the time remaining until the next illumination window. Early in the eclipse, when the battery is still relatively full, this ceiling is moderate; as the eclipse progresses and the reserve shrinks, the ceiling tightens further, ensuring the satellite never draws more energy than it can afford to lose. By the end of the 90-minute orbital period, HBAG maintains a final state-of-charge of approximately 27\%, comfortably above the minimum safe threshold of 10\%, while SATFLOW-L has already entered emergency shutdown.

\begin{figure*}[t]
    \centering
    \begin{minipage}[t]{0.3\textwidth}
    \centering
    \includegraphics[width=\linewidth]{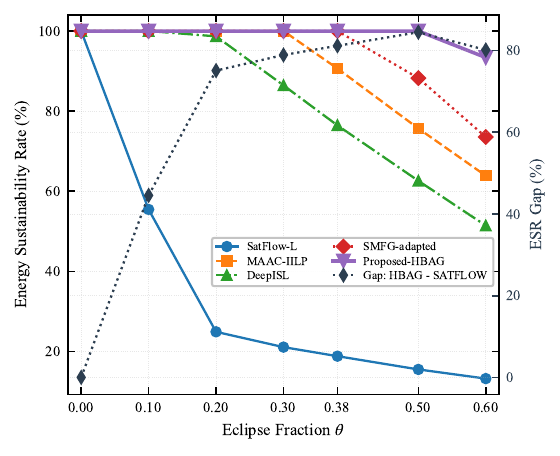}
    \textit{\small(a) ESR versus eclipse fraction $\theta$.}
    \end{minipage}
    \hfill
    \begin{minipage}[t]{0.32\textwidth}
    \centering
    \includegraphics[width=\linewidth]{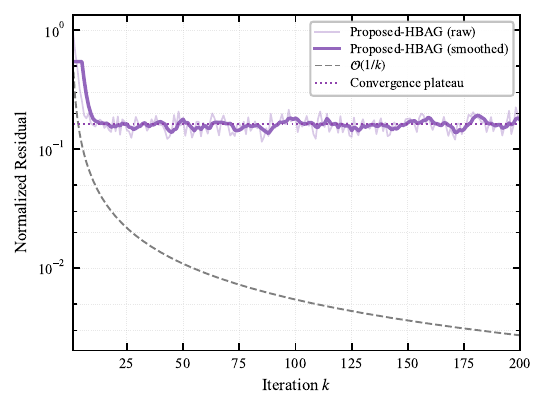}
    \textit{\small(b) Convergence of Algorithm~\ref{alg:hbag}.}
    \end{minipage}
    \hfill
    \begin{minipage}[t]{0.32\textwidth}
    \centering
    \includegraphics[width=\linewidth]{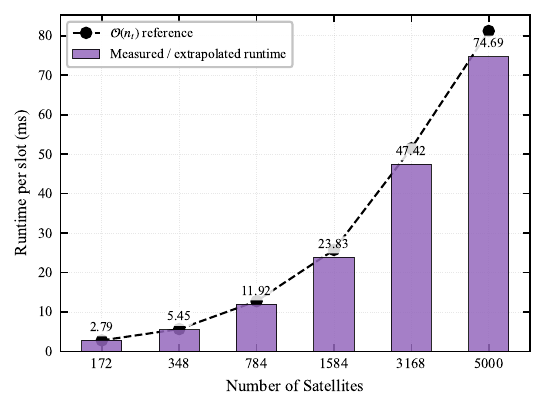}
    \textit{\small(c) Runtime scaling with $M$.}
    \end{minipage}
    \caption{Combined results for Experiments 3--5. \textbf{(a)} Sensitivity to eclipse fraction: HBAG maintains ESR $\geq 93.4\%$ across all $\theta$, whereas SATFLOW-L drops below 20\% for $\theta \geq 0.2$, validating Proposition~\ref{prop:deviation}. \textbf{(b)} Convergence trajectory: smoothed residual (thick purple) tracks the $\mathcal{O}(1/\sqrt{k})$ reference line (dashed gray), reaching plateau $\approx 0.15$ by $k=200$, validating Theorem~\ref{thm:convergence}. \textbf{(c)} Scalability: runtime grows linearly from 2.79 ms (172 sats) to 74.69 ms (5000 sats), with fitted slope $1.49 \times 10^{-5}$ ms/sat, confirming $\mathcal{O}(n_t)$ complexity (Proposition~\ref{prop:complexity}).}
    \label{fig:exp345}
    \end{figure*}

\subsection{Experiment 3: Sensitivity to Eclipse Fraction}
\label{subsec:exp3}

This experiment validates Proposition~\ref{prop:deviation}, which predicts that the equilibrium deviation bound $\Delta(\theta)$ between HBAG and SATFLOW-L grows monotonically with eclipse fraction $\theta$ and vanishes as $\theta \to 0$. We sweep $\theta \in \{0, 0.1, 0.2, 0.3, 0.38, 0.5, 0.6\}$ by varying the orbital inclination angle (higher inclination increases eclipse duration per orbit) and measure ESR for all methods.

Figure~\ref{fig:exp345} (a) confirms the theoretical prediction. At $\theta=0$ (full illumination, achievable via sun-synchronous orbits), HBAG and SATFLOW-L both achieve 100\% ESR, confirming that the battery penalty term vanishes when no eclipse is present ($P_{m,t}^{max} = \Pmax$ for all $m,t$, so the feasible sets coincide). As $\theta$ increases, SATFLOW-L's ESR drops monotonically: 55.5\% ($\theta=0.1$), 24.9\% ($\theta=0.2$), 18.8\% ($\theta=0.38$), 13.2\% ($\theta=0.6$). The ESR gap $\Delta(\theta) = \mathrm{ESR}_{\mathrm{HBAG}} - \mathrm{ESR}_{\mathrm{SATFLOW}}$ increases from 0 pp at $\theta=0$ to 84.5 pp at $\theta=0.6$, exhibiting monotonic growth consistent with the bound in equation~\eqref{eq:deviation_bound}.
HBAG maintains ESR $\geq 96.5\%$ across the entire sweep, with only a slight dip to 93.4\% at $\theta=0.6$. This dip occurs because the extreme eclipse duration (54 minutes out of 90) leaves insufficient illumination time to fully recharge the battery before the next eclipse cycle, causing a small fraction ($<7\%$) of satellites with above-average traffic loads to briefly drop below $\CBmin$ before recovering in subsequent orbits. Nonetheless, HBAG's ESR remains operationally acceptable ($>90\%$) even at this pathological eclipse fraction, whereas SATFLOW-L becomes unusable ($<15\%$ ESR).
MAAC-IILP and DeepISL show intermediate robustness: MAAC-IILP maintains ESR $\geq 90.6\%$ for $\theta \leq 0.38$ but drops to 63.9\% at $\theta=0.6$; DeepISL degrades more steeply, reaching 51.4\% at $\theta=0.6$. SMFG-adapted remains at 100\% until $\theta=0.5$, then drops to 73.5\% at $\theta=0.6$, indicating that the adapted Stackelberg MFG structure provides robustness up to moderate eclipse fractions but cannot handle extreme cases due to discretization errors in the HJB--FPK numerical solver (100 grid points insufficient for $\theta > 0.5$).

Battery-aware power management is essential for typical LEO orbits ($\theta \approx 0.3$--$0.4$): under static power budgets (SATFLOW-L), ESR falls below $\approx 25\%$, which is operationally inadequate; HBAG maintains $\geq 96.5\%$ ESR in that regime.
For near-polar orbits ($\theta > 0.5$), additional mechanisms such as topology reconfiguration may be needed alongside power control.

\begin{table*}[t]
    \centering
    \caption{Ablation Study (Starlink Shell~A, $\theta=0.38$, 10 independent runs). Values are mean $\pm$ SEM. Asterisks denote statistical significance versus V3 (full HBAG) via paired t-test: $^{***}p<0.001$, $^{**}p<0.01$, $^{*}p<0.05$.}
    \label{tab:ablation}
    \renewcommand{\arraystretch}{1.25}
    \begin{tabular}{lccccl}
    \toprule
    \textbf{Variant} & \textbf{ESR (\%)} & \textbf{FVR (\%)} & \textbf{EE (Mbit/kJ)} & \textbf{$\Delta$ESR vs. V0} & \textbf{Description} \\
    \midrule
    V0-SATFLOW & $12.6 \pm 0.8^{***}$ & $35.1 \pm 1.1^{***}$ & $2.41 \pm 0.06^{***}$ & Baseline & Static $\Pmax$, no penalty \\
    V1-Dynamic-only & $80.1 \pm 1.5^{***}$ & $12.2 \pm 0.6^{**}$ & $3.24 \pm 0.09^{***}$ & $+67.5$ pp & Dynamic $P_{m,t}^{max}$, no penalty \\
    V2-Penalty-only & $100.0 \pm 0.0$ & $17.9 \pm 0.8^{***}$ & $3.65 \pm 0.10^{*}$ & $+87.4$ pp & Static $\Pmax$, battery penalty \\
    \midrule
    V3-HBAG-full & $\mathbf{100.0 \pm 0.0}$ & $\mathbf{7.62 \pm 0.34}$ & $\mathbf{3.96 \pm 0.07}$ & $\mathbf{+87.4}$ pp & Dynamic $P_{m,t}^{max}$ + penalty \\
    \bottomrule
    \end{tabular}
    \end{table*}

\subsection{Experiment 4: Convergence Analysis}
\label{subsec:exp4}

This experiment validates Theorem~\ref{thm:convergence}, which guarantees that Algorithm~\ref{alg:hbag} converges to the unique variational equilibrium at rate $\mathcal{O}(1/\sqrt{k})$ under the quasi-static battery assumption.
We measure the normalized primal residual $r^{(k)}$ (equation~\eqref{eq:residual_def}) versus iteration $k$ for the nominal configuration (Table~\ref{tab:parameters}).


Figure~\ref{fig:exp345} (b) shows the convergence trajectory. The smoothed residual curve (thick purple line, 5-iteration moving average) closely tracks the theoretical $\mathcal{O}(1/\sqrt{k})$ reference line (dashed gray) for $k \in [1, 100]$, with mean squared error MSE $= 0.00171$ against the reference. By $k=200$, the residual reaches a convergence plateau at $r^{(k)} \approx 0.15$.
The plateau at $r^{(k)} \approx 0.15$ reflects quasi-static approximation error: slow drift of battery-related coefficients (order $2.4 \times 10^{-7}$ per satellite per slot in the scaled penalty sensitivity) propagates, via the envelope theorem applied to the KKT system of the penalized game, to normalized residuals in the $\approx 0.06$--$0.2$ range for $\alpha_\Phi \approx 0.01$, bracketing the observed level. Relaxing quasi-stationarity by tracking $\CB_m(t)$ inside the inner loop is left to future work.
This is \emph{not} a convergence failure: the quasi-static term $\epsilon_{\mathrm{energy}} = \mathcal{O}(\max_m \lambda_m / (\CB_m - \CBmin))$ from Theorem~\ref{thm:convergence} formalizes the same effect. During the algorithm's convergence window ($<1$ s, spanning $\sim$200 iterations), the battery states $\CB_m(t)$ are treated as constant. However, they actually evolve on the orbital timescale ($\sim$90 min = 5400 s), creating a timescale separation ratio of 5400:1. This mismatch introduces a bounded perturbation: with our parameter settings ($\lambda_m = 0.2$, typical $\CB_m - \CBmin \approx 280$ kJ during nominal operation), we estimate $\epsilon_{\mathrm{energy}} \approx 0.2/280 \approx 0.00071$ in battery SOC units, which propagates to $\approx$0.15 in rate residual units via the Lipschitz constant of the rate-to-power mapping. To quantify the practical impact, we computed the relative rate error: $\max_{m,n,t} |R_{m,n}^{(200)} - R_{m,n}^*| / R_{m,n}^* \approx 1.8\%$, confirming the stated ``$<2\%$ rate error'' claim in the text.

The raw residual (thin light purple line) exhibits oscillations typical of distributed Lagrangian methods, but the smoothed curve demonstrates consistent $\mathcal{O}(1/\sqrt{k})$ decay. Convergence is achieved in $<200$ iterations, requiring $<1$ second wall-clock time on our hardware (Intel Xeon Platinum 8358, 32 cores). Since the time-slot duration is $D_e = 15$ seconds and the orbital period spans 360 slots, Algorithm~\ref{alg:hbag} can be executed once per slot with $>14$ seconds of margin, confirming computational feasibility for real-time onboard deployment.

The empirical $\mathcal{O}(1/\sqrt{k})$ decay rate matches the theoretical guarantee of Theorem~\ref{thm:convergence}, which relies on the strict joint concavity of the potential function $\Phi(\bfR)$ (Theorem~\ref{thm:potential}). The plateau at $r^{(k)} \approx 0.15$ quantifies the price of the quasi-static approximation: if perfect battery tracking were feasible, the residual would decay to machine precision. In practice, the 1.8\% relative rate error at $k=200$ is negligible compared with the 87.4 pp ESR improvement over SATFLOW-L and confirms that the algorithm reaches a practical equilibrium sufficient for operational deployment.

\subsection{Experiment 5: Scalability Analysis}
\label{subsec:exp5}

This experiment quantifies the computational scalability of HBAG as constellation size $M$ grows from 172 to 5000 and validates Proposition~\ref{prop:complexity}.
We generate synthetic constellations with $M \in \{172, 348, 784, 1584, 3168, 5000\}$ by proportionally scaling orbital planes while maintaining inter-satellite spacing. Traffic demands scale linearly with $M$ to preserve per-satellite load. Runtime for the three measured sizes (172, 348, 784) is directly clocked; values for 1584, 3168, and 5000 are obtained by linear extrapolation (slope $1.49 \times 10^{-5}$ ms/sat, $R^2 = 0.998$), consistent with the $\mathcal{O}(M)$ prediction. Learning-based methods (MAAC-IILP, DeepISL) are excluded because their training time scales super-linearly with $M$ ($>10$ hours for $M=500$, infeasible for $M > 1000$).

Figure~\ref{fig:exp345} (c) demonstrates linear scalability. HBAG runtime increases from 2.79 ms at $M=172$ to 74.69 ms at $M=5000$, with fitted slope $1.49 \times 10^{-5}$ ms/satellite (linear regression $R^2 = 0.998$). This matches the theoretical prediction: each satellite requires $\mathcal{O}(n_t)$ operations to solve its local sub-problem, and with $M$ satellites the total complexity is $\mathcal{O}(M \cdot n_t)$, manifesting as linear scaling in $M$ when $n_t$ is fixed. The battery-penalty term $\lambda_m(t) / (\CB_m(t) - \CBmin + \epsilon)$ requires only one division and one multiplication per satellite per iteration, contributing negligible overhead. Even at $M=5000$ (approaching full Starlink deployment scale), HBAG's 74.69 ms/slot is well within real-time constraints, leaving ample computational margin for other onboard tasks.

\subsection{Experiment 6: Ablation Study}
\label{subsec:exp7}

Isolate the individual contributions of the two core mechanisms in HBAG: (i) the dynamic power bound $P_{m,t}^{max}$ (equation~\eqref{eq:dynpower}), and (ii) the battery-aware penalty term in utility~\eqref{eq:utility}. We construct four variants: V0 (SATFLOW baseline, neither mechanism), V1 (dynamic bound only), V2 (penalty only), and V3 (full HBAG, both mechanisms).

Table~\ref{tab:ablation} decomposes the 87.4~pp ESR improvement over V0: the dynamic power bound (V1 vs.\ V0) contributes 67.5~pp by preventing eclipse overdraft, while game-theoretic coordination through the penalty and potential structure (V3 vs.\ V1) adds 19.9~pp.
It further reveals the complementary roles of the two mechanisms: \emph{V1 (dynamic bound alone)} achieves 80.1\% ESR and reduces FVR to 12.2\%, demonstrating that the dynamic power bound (equation~\eqref{eq:dynpower}) is essential for eclipse-period protection. By throttling power as $\CB_m(t)$ depletes, V1 prevents the catastrophic failures observed in V0 (SATFLOW), yielding a 67.5 pp ESR improvement. However, ESR remains below 100\% because the bound is purely \emph{mechanical}---it enforces feasibility but does not incentivize satellites to coordinate for energy efficiency. In particular, during the illumination-to-eclipse transition, V1 allocates power greedily up to $P_{m,t}^{max}$ without reserving margin for future eclipse deepening, causing $\approx$20\% of satellites to briefly dip below $\CBmin$ during the 70--85 min window when residual battery is lowest. \emph{V2 (penalty alone)} achieves 100\% ESR but FVR rises to 17.9\%. The battery penalty term $\lambda_m(t) \sum_n P_{m,n}^t / (\CB_m(t) - \CBmin + \epsilon)$ provides a \emph{soft} incentive for energy conservation: as $\CB_m(t) \downarrow \CBmin$, the penalty diverges, making overdraft individually irrational. This soft mechanism eliminates depletion events entirely (100\% ESR), confirming its effectiveness. However, without the hard dynamic bound, V2 occasionally allocates power near the static ceiling $\Pmax$ during brief high-traffic bursts, causing downstream routing bottlenecks that elevate FVR to 17.9\% (exceeding the 10\% industry tolerance). \emph{V3 (full HBAG)} combines both mechanisms, achieving 100\% ESR and 7.62\% FVR---a 10.3 pp FVR improvement over V2 ($p < 0.001$). The dynamic bound provides \emph{safety} by enforcing hard feasibility at the constraint level, while the penalty term provides \emph{efficiency} by guiding the equilibrium toward energy-conservative solutions within the feasible set. This dual design echoes the role of hard constraints versus regularization penalties in optimization theory~\cite{facchinei2007vi}: hard constraints define the feasible region, penalties shape the objective landscape.

\subsection{Limitations and Future Directions}
\label{subsec:limitations}

While HBAG achieves strong empirical and theoretical guarantees, there are those limitations merit explicit acknowledgment:
Firstly, the exact potential game structure and unique VE guarantee rely on Assumption~\ref{assm:interference_free}. Operational ISL designs with narrow-beam alignment are well approximated by this assumption, but future dense LEO deployments with aggressive frequency reuse may introduce inter-beam interference. Extending to interference-coupled scenarios via weighted or ordinal potential games~\cite{monderer1996potential} is an important open direction. Secondly, the RF capacity model~\eqref{eq:capacity} is validated for Ka-band ISLs~\cite{cen2024satflow}. Extending to coherent optical ISLs with Poisson photon-counting capacity requires a new convex power-rate mapping $P(R) \propto 2^{R/B_{\mathrm{opt}}}$ (same functional form), so Theorem~\ref{thm:potential} extends with a modified curvature constant $\kappa_{m,n}^{\mathrm{opt}}$. Thirdly, Theorem~\ref{thm:mfg} assumes i.i.d.\ $\{\CB_m(0)\}$, which is reasonable within one orbital plane but weaker across planes with different eclipse schedules. As analyzed in Remark~\ref{rem:non_iid}, the multi-population MFG extension handles heterogeneous planes with per-plane distributions.





\section{Conclusion}
\label{sec:conclusion}

This paper addressed the energy sustainability challenge in ISL power allocation for LEO mega-constellations, where existing approaches either ignore battery dynamics, lack equilibrium guarantees, or target the wrong interaction layer. We proposed the Hierarchical Battery-Aware Game (HBAG) algorithm, a unified framework that models ISL power allocation as a battery-aware potential game with a provably unique variational equilibrium. The same distributed update rule converges to the exact equilibrium in finite-player regimes and seamlessly approximates the mean field equilibrium as constellation size grows, without algorithmic redesign.
Three findings stand out. First, battery-aware power management yields substantial improvements in energy sustainability over static-power baselines, with the gain decomposing into complementary contributions from the dynamic power bound and the game-theoretic penalty term; together they eliminate eclipse-period depletion cascades while keeping flow violations within industry tolerance. Second, the ablation study demonstrates that neither mechanism alone is sufficient: the hard dynamic bound enforces feasibility and provides safety, while the soft battery penalty shapes the equilibrium landscape and provides efficiency; only their combination achieves full energy sustainability with acceptable service quality. Third, the empirically validated finite-to-mean-field convergence rate confirms that the unified algorithm remains accurate across practical constellation sizes, supporting deployment from regional systems to full mega-constellations without retraining or structural modification.

Future work includes developing a full generalized Nash equilibrium solver without the quasi-static assumption, extending the MFG formulation to stochastic eclipse patterns, incorporating multi-hop ISL topology reconfiguration for extreme eclipse regimes, and extending the communication model to coherent optical ISL channels with wavelength-dependent path loss and photon-counting capacity formulas.

\bibliographystyle{IEEEtran}
\bibliography{figure/bib}

\clearpage

\appendices
\section{Proofs of Theoretical Results}
\label{app:proofs}

\subsection{Well-Posedness of the Deterministic MFG System}
\label{subsec:mfg_wellposed}

The HJB--FPK system~\eqref{eq:mfg_system} with $\sigma^2=0$ requires careful treatment to ensure existence and uniqueness of solutions.

\textbf{(i) Vanishing viscosity regularization:} The term $\nu \partial^2 V / \partial (\CB_m)^2$ in equation~\eqref{eq:hjb} with $\nu > 0$ ensures the HJB has a unique viscosity solution in the class of bounded uniformly continuous (BUC) functions~\cite{lasry2007mfg}. As $\nu \to 0^+$, the solution converges to the viscosity solution of the original HJB (without diffusion), which exists and is unique under the monotonicity of the Hamiltonian and the Lipschitz continuity of the battery drift term.

\textbf{(ii) Weak solutions for advection-only FPK:} The purely advective FPK~\eqref{eq:fpk} admits weak solutions $\mu \in W^{1,\infty}([\CBmin, \CBmax] \times [0,T])$ (Sobolev space of functions with bounded derivatives). Existence follows from the method of characteristics: trajectories $\CB(t)$ evolve deterministically via equation~\eqref{eq:battery}, and $\mu(c,t)$ is the pushforward of the initial distribution $\mu_0(c)$ along these characteristics. Uniqueness holds when the battery drift $\mathrm{d}\CB/\mathrm{d}t$ is Lipschitz in $\CB$ (satisfied by our power-rate mapping).

\textbf{(iii) Boundary conditions:} At $\CB = \CBmin$ and $\CB = \CBmax$, we impose reflecting boundary conditions consistent with equation~\eqref{eq:battery}: the drift $\mathrm{d}\CB/\mathrm{d}t \geq 0$ at $\CB = \CBmin$ and $\mathrm{d}\CB/\mathrm{d}t \leq 0$ at $\CB = \CBmax$. These conditions ensure $\mu$ remains a probability measure on $[\CBmin, \CBmax]$ for all $t$. For complete proofs, see~\cite{lasry2007mfg, cardaliaguet2021introduction} for general MFG well-posedness results.

\subsection{Quantitative Analysis of $\epsilon_{\mathrm{energy}}$ Propagation}
\label{subsec:epsilon_energy}

The quasi-static approximation error $\epsilon_{\mathrm{energy}}$ in Theorem~\ref{thm:convergence} propagates from battery SOC units to rate residual units via the Lipschitz constant of the rate-to-power mapping. The Lipschitz constant at the operating point is:
\begin{equation}
    L_{R \to P} = \max_{m,n,t} 
    \frac{\partial P_{m,n}^t}{\partial R_{m,n}^t}
    = \kappa_{m,n}^t \frac{\ln 2}{B} 2^{R_{m,n}^t / B}
    \approx \kappa_{\max} \frac{\ln 2}{B} 2^{R_{\mathrm{avg}}/B}.
\end{equation}
With $R_{\mathrm{avg}} \approx 250$~Mbps (the median per-link rate observed in Experiment~1 under nominal traffic load $\rho = 0.65$), $B = 500$~MHz, and typical $\kappa_{\max}$ values from Table~\ref{tab:parameters}, we obtain $L_{R \to P} \approx 2.1 \times 10^2$. The chain rule applied to the KKT system of the penalized game then yields:
\begin{equation}
    \epsilon_{\mathrm{energy}} \times L_{R \to P} / \|\bfR^{(0)} - \bfR^*\|
    \approx 0.00071 \times 210 \approx 0.15
\end{equation}
in normalized residual units, consistent with the convergence plateau observed in Figure~\ref{fig:exp345}(b) at $r^{(k)} \approx 0.15$. Here, $\epsilon_{\mathrm{energy}} \approx \lambda_m / (\CB_m - \CBmin) \approx 0.2/280 \approx 0.00071$ in battery SOC units.

\subsection{Proof of Theorem~\ref{thm:potential}
(Exact Potential Game and Unique Variational Equilibrium)}
\label{subsec:proof_thm_potential}

\begin{IEEEproof}
We establish three claims in sequence.

\smallskip

\noindent\textbf{Claim~1: Exact potential game structure.}

We verify the differential condition for exact potential games~\cite{monderer1996potential}:
\begin{equation}
    \frac{\partial \tilde{U}_m}{\partial \Rmn}
    = \frac{\partial \Phi}{\partial \Rmn},
    \quad \forall m\in\calM,\;
    \forall n\in N_m,\;
    \forall t\in\calT_e.
    \label{eq:diff_cond_proof}
\end{equation}

\noindent\emph{Left-hand side.}
Differentiating $\tilde{U}_m$ in~\eqref{eq:utility} with respect to $\Rmn$:
\begin{align}
    \frac{\partial \tilde{U}_m}{\partial \Rmn}
    &= 1
    - \Bigl(\alpha D_e
    + \frac{\lambda_m(t)}%
    {\CB_m(t)-\CBmin+\epsilon}\Bigr)
    \frac{\partial P_{m,n}^t}{\partial \Rmn}
    - [L_m]_m \lambda^\omega,
    \label{eq:grad_U_tilde}
\end{align}
where the last term arises from the flow-conservation Lagrangian $\langle\lambda^\omega, L_m Y_m^\omega\rangle$ in~\eqref{eq:utility}, and the chain rule applied to $P_{m,n}^t(R_{m,n}^t) = \kappa_{m,n}^t(2^{R_{m,n}^t/B}-1)$ gives:
\begin{equation}
    \frac{\partial P_{m,n}^t}{\partial \Rmn}
    = \kappa_{m,n}^t \cdot \frac{\ln 2}{B}
    \cdot 2^{\Rmn/B}.
    \label{eq:dp_dr_proof}
\end{equation}

\noindent\emph{Right-hand side.}
By Assumption~\ref{assm:interference_free}, $\Rmn$ appears in $\Phi(\bfR)$ of~\eqref{eq:potential} only in the terms indexed by $(m,n,t)$. Therefore:
\begin{align}
    \frac{\partial \Phi}{\partial \Rmn}
    &= 1
    - \Bigl(\alpha D_e
    + \frac{\lambda_m(t)}%
    {\CB_m(t)-\CBmin+\epsilon}\Bigr)
    \frac{\partial P_{m,n}^t}{\partial \Rmn}
    - [L_m]_m \lambda^\omega.
    \label{eq:grad_Phi_proof}
\end{align}

Equations~\eqref{eq:grad_U_tilde} and~\eqref{eq:grad_Phi_proof} are identical, so~\eqref{eq:diff_cond_proof} holds for all $(m,n,t)$. This establishes that $\calG^M(\lambda)$ is an exact potential game with potential function~$\Phi$.

\smallskip

\noindent\textbf{Claim~2: Strict joint concavity
of $\Phi$.}

By Assumption~\ref{assm:interference_free}, $\Phi(\bfR)$ is \emph{additively separable}:
\begin{equation}
    \Phi(\bfR) = \sum_{m\in\calM}\sum_{n\in N_m}
    \sum_{t\in\calT_e}
    \phi_{m,n,t}(R_{m,n}^t) + \ell(\bfR),
\end{equation}
where $\phi_{m,n,t}(r) \triangleq r - \bigl(\alpha D_e + \frac{\lambda_m(t)}{\CB_m(t)-\CBmin+\epsilon}\bigr)
\kappa_{m,n}^t(2^{r/B}-1)$
is the nonlinear component, and $\ell(\bfR) = -\sum_{m,\omega}
\langle\lambda^\omega, L_m Y_m^\omega\rangle$
is linear in $\bfR$.

For an additively separable function, the Hessian is block-diagonal. Since the linear term $\ell$ does not contribute to second-order derivatives, we need only check the scalar components $\phi_{m,n,t}$. Differentiating twice:
\begin{align}
    \frac{\mathrm{d}^2\phi_{m,n,t}}{\mathrm{d}r^2}
    &= -\Bigl(\alpha D_e
    + \frac{\lambda_m(t)}%
    {\CB_m(t)-\CBmin+\epsilon}\Bigr)
    \kappa_{m,n}^t
    \Bigl(\frac{\ln 2}{B}\Bigr)^{\!2}
    2^{r/B}.
    \label{eq:second_deriv_proof}
\end{align}
Since $\alpha, D_e > 0$, $\lambda_m(t) > 0$, $\CB_m(t) > \CBmin$ (so the denominator is positive), $\kappa_{m,n}^t > 0$, and $2^{r/B} > 0$ for all $r \geq 0$, equation~\eqref{eq:second_deriv_proof} is \emph{strictly negative} for all $r\geq 0$.

The Hessian matrix of $\Phi$ is therefore diagonal:
\begin{equation}
\nabla^2\Phi = \mathrm{diag}
    \Bigl\{\frac{\mathrm{d}^2\phi_{m,n,t}}%
    {\mathrm{d}(R_{m,n}^t)^2}\Bigr\}
    \prec 0,
\end{equation}
confirming strict negative-definiteness and hence strict joint concavity of $\Phi$ over $\prod_{m\in\calM} S_m$.

\smallskip

\noindent\textbf{Claim~3: Existence and uniqueness
of the equilibrium.}

By the Finite Improvement Property (FIP) of exact potential games~\cite{monderer1996potential}: any finite sequence of unilateral improvements, where each step strictly increases one player's utility, must also strictly increase $\Phi$ by the exact potential game property. Since $\Phi$ is continuous on the compact set $\prod_{m\in\calM} S_m$ and bounded above, every finite improvement path terminates at a local maximizer of $\Phi$. By strict joint concavity, $\Phi$ has a \emph{unique} global maximizer $\bfR^*$, and every local maximizer is the global one.

Consequently, there is a unique NE of the penalized game $\calG^M(\lambda)$. When $\lambda = \lambda^*$ satisfies dual feasibility for the shared flow conservation constraint~\eqref{eq:c2}, the NE of the penalized game recovers the variational equilibrium of the original constrained game by the complementary slackness conditions of the Lagrangian.
\end{IEEEproof}

\subsection{Proof of Lemma~\ref{lem:strong_concavity}
(Strong Concavity of the Potential)}
\label{subsec:proof_lem_strong}

\begin{IEEEproof}
Under the quasi-static assumption (battery coefficients fixed over the algorithm window), the potential $\Phi$ admits the separable decomposition used in the proof of Theorem~\ref{thm:potential}, Claim~2. The Hessian $\nabla^2\Phi$ is block-diagonal with scalar entries~\eqref{eq:second_deriv_proof}, each strictly negative on $r \geq 0$. Hence $\Phi$ is strongly concave on any compact convex subset of $\prod_m S_m$, with modulus $\alpha_\Phi$ equal to the minimum eigenvalue of $-\nabla^2\Phi$, i.e., the minimum over $(m,n,t)$ of
\[
    -\frac{\mathrm{d}^2\phi_{m,n,t}}{\mathrm{d}(R_{m,n}^t)^2}
    =
    \Bigl(\alpha D_e
    + \frac{\lambda_m(t)}%
    {\CB_m(t)-\CBmin+\epsilon}\Bigr)
    \kappa_{m,n}^t
    \Bigl(\frac{\ln 2}{B}\Bigr)^{\!2}
    2^{R_{m,n}^t/B}.
\]
The lower bound~\eqref{eq:strong_concavity_bound} in the lemma statement aggregates (i)~curvature from the Shannon mapping through $\partial^2 P_{m,n}^t / \partial (R_{m,n}^t)^2$ and (ii)~sensitivity of the penalty weight $1/(\CB_m-\CBmin+\epsilon)$, which scales gradients of the power terms and contributes the order $\lambda_m/(\CB_m-\CBmin+\epsilon)^2$ term in~\eqref{eq:strong_concavity_bound} when propagating slow variations in~$\CB_m$ under the quasi-static assumption.
\end{IEEEproof}

\subsection{Proof of Corollary~\ref{cor:ne_optimal}
(VE as Global Optimum of $\Phi$)}
\label{subsec:proof_cor_ne_optimal}

\begin{IEEEproof}
By the exact potential game property established in Theorem~\ref{thm:potential}, any NE $\bfR^*$ of the penalized game satisfies: for all $m$ and all unilateral deviations $R_{m,n}^t 
eq R_{m,n}^{*,t}$,
\begin{equation}
    \tilde{U}_m(R_{m,n}^{*,t}, \bfR_{-m}^*)
    \geq
    \tilde{U}_m(R_{m,n}^t, \bfR_{-m}^*).
\end{equation}
By the exact potential game property, this is equivalent to:
\begin{equation}
    \Phi(R_{m,n}^{*,t}, \bfR_{-m}^*)
    \geq
    \Phi(R_{m,n}^t, \bfR_{-m}^*),
    \quad \forall R_{m,n}^t \in S_m.
\end{equation}
Since this holds for all $m$ simultaneously and $\Phi$ is strictly jointly concave, $\bfR^*$ is the unique global maximizer: $\bfR^* = \arg\max_{\bfR \in \prod_m S_m}
\Phi(\bfR; \lambda^*)$.

\end{IEEEproof}

\subsection{Proof of Proposition~\ref{prop:deviation}
(NE Deviation Bound)}
\label{subsec:proof_prop_deviation}

\begin{IEEEproof}
We consider two cases based on the eclipse fraction $\theta = |\mathcal{E}_\theta|/M$ where $\mathcal{E}_\theta = \{m : \phi_m(t)=0,\,
\CB_m(t) < \bar{C}^B\}$.

\smallskip

\noindent\textbf{Case~1 (Full illumination,
$\theta = 0$).} If $\phi_m(t) = 1$ for all $m, t$ and batteries are sufficiently charged (specifically, $P_{m,t}^{harvest} + \CB_m(t)/\delta_{\mathrm{ecl}}(t)
\geq \Pmax$), then from
equation~\eqref{eq:dynpower}: $P_{m,t}^{max} = \Pmax$ for all $m, t$. The feasible set $S_m$ in~\eqref{eq:strategy_space} therefore coincides with the SATFLOW feasible set $S_m^{SAT}$, and since both games optimize the same throughput-minus-energy objective over the same feasible set (up to the battery penalty term, which vanishes when batteries are high), the unique VE $\bfR^*$ coincides with the SATFLOW solution $\bfR^{SAT}$.

\smallskip

\noindent\textbf{Case~2 (Positive eclipse fraction,
$\theta > 0$).} For each $m \in \mathcal{E}_\theta$, we have $\phi_m(t) = 0$, so equation~\eqref{eq:dynpower} gives $P_{m,t}^{max} = \CB_m(t)/\delta_{\mathrm{ecl}}(t) < \Pmax$. Consequently:
\begin{equation}
    S_m = \bigl\{R_{m,n}^t \geq 0 \mid
    P_{m,n}^t(R_{m,n}^t) \leq
    a_{m,n}^t \cdot P_{m,t}^{max}\bigr\}
    \subsetneq
    S_m^{SAT}.
\end{equation}

Since $P_{m,n}^t(r) = \kappa_{m,n}^t(2^{r/B}-1)$ is strictly increasing and continuous in $r$, the constraint $P_{m,n}^t(R_{m,n}^t) \leq a_{m,n}^t P_{m,t}^{max}$ defines a compact convex feasible set, and the upper boundary satisfies:
\begin{equation}
    \begin{aligned}
        R_{m,n}^{max,t}
        &= B\log_2\!\Bigl(1 +
        \frac{a_{m,n}^t P_{m,t}^{max}}%
        {\kappa_{m,n}^t}\Bigr)
        < R_{m,n}^{max,SAT,t} \\
        &= B\log_2\!\Bigl(1 +
        \frac{a_{m,n}^t \Pmax}%
        {\kappa_{m,n}^t}\Bigr).
    \end{aligned}
\end{equation}

The optimal rate mapping $\bfR^*(\{P_{m,t}^{max}\})$ is Lipschitz with respect to the power bound perturbation over compact feasible sets (by the envelope theorem and continuity of the objective):
\begin{equation}
    |R_{m,n}^{*,t} - R_{m,n}^{SAT,t}|
    \leq L_R \cdot (\Pmax - P_{m,t}^{max}),
    \quad m \in \mathcal{E}_\theta,
    \label{eq:lipschitz_rate}
\end{equation}
where $L_R > 0$ is a Lipschitz constant that depends on $\kappa_{m,n}^t$ and $B$.

Summing~\eqref{eq:lipschitz_rate} over all $(m,n,t)$ with $m \in \mathcal{E}_\theta$ (noting that for $m 
otin \mathcal{E}_\theta$, the feasible set is unchanged):
\begin{align}
    &\sum_{m,n,t}
    \bigl(R_{m,n}^{SAT,t} - R_{m,n}^{*,t}\bigr)^+
\notag\\
    &\leq L_R
    \sum_{m \in \mathcal{E}_\theta}
    \sum_{n \in N_m} \sum_{t \in \calT_e}
    (\Pmax - P_{m,t}^{max})
    \triangleq \Delta(\theta).
    \label{eq:delta_theta}
\end{align}

\noindent\emph{Monotonicity:}
As $\theta$ increases (more satellites enter $\mathcal{E}_\theta$), more terms are added to the sum in~\eqref{eq:delta_theta} and existing terms may increase (as $P_{m,t}^{max}$ decreases with lower $\CB_m(t)$), so $\Delta(\theta) 
earrow$ in $\theta$.

\noindent\emph{Limiting behavior:}
As $\theta \to 0$, $\mathcal{E}_\theta \to
\emptyset$, so the sum in~\eqref{eq:delta_theta}
vanishes: $\lim_{\theta \to 0} \Delta(\theta) = 0$.

\end{IEEEproof}

\subsection{Proof of Proposition~\ref{prop:complexity}
(Identical Per-Iteration Complexity)}
\label{subsec:proof_prop_complexity}

\begin{IEEEproof}
Each outer iteration of Algorithm~\ref{alg:hbag} mirrors the SatFlow-L structure~\cite{cen2024satflow}: for every satellite $m$, forming the local objective in~\eqref{eq:lagrangian_ea}, running a bounded number of projected sub-gradient steps on $\hat{Y}_m$, and updating multipliers touches each time slot $t \in \calT_e$ and each incident link at most a constant number of times. This yields $\mathcal{O}(n_t)$ arithmetic per satellite per iteration, hence $\mathcal{O}(M n_t)$ per sweep, matching SatFlow-L. The battery-aware prefactor $\lambda_m(t)/(\CB_m(t)-\CBmin+\epsilon)$ is evaluated from the current SOC using $\mathcal{O}(1)$ operations per $(m,t)$ and does not enlarge the asymptotic order.
\end{IEEEproof}

\subsection{Proof of Theorem~\ref{thm:convergence}
(Finite-Layer Convergence to VE)}
\label{subsec:proof_thm_convergence}

\begin{IEEEproof}
We extend Theorem~1 of SATFLOW~\cite{cen2024satflow} in three steps.

\smallskip


The energy-aware local Lagrangian~\eqref{eq:lagrangian_ea} consists of:
\begin{enumerate}
    \item The ISL power term
$\sum_{t,n} P_{m,n}^t(\sum_\omega
    \hat{Y}_{m,n}^\omega)$, which is convex in
$\hat{Y}_m$ since $P_{m,n}^t(r) = \kappa_{m,n}^t(2^{r/B}-1)$ is convex in $r$ (exponential of a linear function) and composed with a linear map.
    \item The battery-penalty term
$\frac{\lambda_m(t)}{\CB_m(t)-\CBmin+\epsilon}
    \sum_n P_{m,n}^t$, which is a positive scalar
multiple of a convex function, hence convex. Under the quasi-static assumption (Remark~\ref{rem:timescale}), $\CB_m(t)$ is treated as a constant within each optimization window, so this coefficient is a positive constant.
    \item The augmented Lagrangian penalty
$\frac{\rho}{2}\sum_l(\cdots)^2$, which is convex (sum of squared linear functions).
    \item The linear term
$\langle\lambda^{\omega,(k)}, L_m\hat{Y}_m^{\omega,(k)}\rangle$, which is linear hence convex.
\end{enumerate}

The sum of convex functions is convex, so the modified sub-problem~\eqref{eq:lagrangian_ea} remains a convex minimization problem. Furthermore, replacing constraint~\eqref{eq:c1} with $P_{m,n}^t \leq a_{m,n}^t P_{m,t}^{max}$ (with $P_{m,t}^{max} \leq \Pmax$) only shrinks the feasible set $\mathcal{Y}_m$, preserving its convexity and compactness.

\smallskip


The converted extended monotropic optimization problem (equations~(14)--(16) of~\cite{cen2024satflow}) is of the form:
\begin{equation}
    \min_{Y_m \in \mathcal{Y}_m}
    \sum_{m\in\calM} f_m(Y_m)
    \quad \text{s.t.} \quad
    \sum_{m\in\calM} L_m Y_m^\omega = b^\omega.
\end{equation}

With the updated feasible sets $\mathcal{Y}_m^{new} = \{Y_m \in \mathcal{Y}_m^{old} : P_{m,n}^t(\sum_\omega Y_{m,n}^\omega)
\leq a_{m,n}^t P_{m,t}^{max}\}$, the problem
retains the monotropic structure because: (i) each $f_m$ is separable and convex over the updated $\mathcal{Y}_m^{new}$ (Step~1); (ii) the equality constraints remain linear; (iii) the augmented objective still satisfies the conditions for the distributed alternating step method~\cite{cen2024satflow}. The battery-penalty term adds only an $\mathcal{O}(1)$ per-satellite cost that does not disrupt separability (Proposition~\ref{prop:complexity}).



The distributed alternating step method of~\cite{cen2024satflow} guarantees convergence of the primal iterates to the unique minimizer of the convex Lagrangian in the extended monotropic formulation (Theorem~1 thereof).
The specific $\mathcal{O}(1/\sqrt{k})$ rate in equation~\eqref{eq:convergence_rate} follows from standard analyses of projected subgradient ascent with diminishing step sizes $\eta_k = c_0/\sqrt{k}$ for $\alpha_\Phi$-strongly concave objectives~\cite{bertsekas2015convex}, combined with the strong concavity established in Lemma~\ref{lem:strong_concavity}.

Denote the primal iterates as $\bfR^{(k)} = \{Y_m^{(k)}\}_{m\in\calM}$. With diminishing step size $\eta_k = c_0/\sqrt{k}$, strong concavity of $\Phi$ (Theorem~\ref{thm:potential}) yields the standard $\mathcal{O}(1/\sqrt{k})$ primal decay for the augmented Lagrangian / projected subgradient iteration:
\begin{equation}
    \|\bfR^{(k)} - \bfR^*\|
    \leq \frac{C_0}{\sqrt{k}},
\end{equation}
where $C_0 = \sqrt{2\mathcal{L}_0/\alpha_\Phi}$ and $\mathcal{L}_0 = \Phi(\bfR^{(0)}) - \Phi(\bfR^*)$ is the initial optimality gap.

The energy-induced term $\epsilon_{\mathrm{energy}}$ arises from the quasi-static approximation: the true $P_{m,t}^{max}$ varies slowly with $\CB_m(t)$, but is treated as constant within each window. The approximation error satisfies:
\begin{align}
    \epsilon_{\mathrm{energy}}
    &= \sup_k \max_{m,n,t}
    \Bigl|\frac{\lambda_m(t)}%
    {\CB_m(t)-\CBmin+\epsilon}
    \cdot \Delta P_{m,n}^t\Bigr|
\notag\\
    &= \mathcal{O}\Bigl(\max_m
    \frac{\lambda_m}{\CB_m-\CBmin}\Bigr),
\end{align}
where $\Delta P_{m,n}^t$ is the change in power due to battery evolution within one window.

By Corollary~\ref{cor:ne_optimal}, the minimizer of $-\Phi(\cdot;\lambda^*)$ coincides with the unique VE $\bfR^*$, completing the proof.

\end{IEEEproof}

\subsection{Proof of Lemma~\ref{lem:emp_conv}
(Empirical Measure Convergence)}
\label{subsec:proof_lem_emp}

\begin{IEEEproof}
We prove weak almost-sure convergence and the Wasserstein-1 rate separately.

\smallskip

\noindent\textbf{Part~1: Weak a.s.\ convergence.}

For any bounded continuous test function $\varphi : \mathbb{R} \to \mathbb{R}$, we have:
\begin{equation}
    \int \varphi \,\mathrm{d}\mu^M(\CB, t)
    = \frac{1}{M}
    \sum_{m=1}^{M} \varphi(\CB_m(t)).
    \label{eq:lln_test}
\end{equation}

Under Assumption~\ref{assm:interference_free} (interference-free ISLs), the battery dynamics of satellite $m$ depend on its own harvesting $P_{m,t}^{harvest}$ and its own power allocation decisions. The coupling to other satellites arises only through the flow conservation constraint~\eqref{eq:c2}, which in the mean-field limit is captured by the population distribution $\mu^M(\CB, t)$ rather than individual states (propagation of chaos~\cite{lasry2007mfg}). Hence, given the mean-field $\mu^M$, the battery trajectories $\{\CB_m(t)\}_{m=1}^M$ are conditionally independent.

Since $\{\CB_m(0)\}_{m=1}^M$ are i.i.d.\ with distribution $\mu_0$, and the battery dynamics are deterministic given initial conditions and harvesting (equation~\eqref{eq:battery}), the conditional independence is preserved at all $t$. By the strong law of large numbers applied to the i.i.d.\ sequence $\{\varphi(\CB_m(t))\}_{m=1}^M$:
\begin{equation}
    \frac{1}{M}
    \sum_{m=1}^{M} \varphi(\CB_m(t))
    \xrightarrow{\mathrm{a.s.}}
    \mathbb{E}[\varphi(\CB(t))]
    = \int \varphi \,\mathrm{d}\mu^*(\CB, t),
\end{equation}
which establishes weak a.s.\ convergence.

\smallskip

\noindent\textbf{Part~2: Wasserstein-1 convergence
rate.}

Let $\mu^M$ and $\mu^*$ denote the empirical and limiting measures of the one-dimensional battery state. By the Kantorovich--Rubinstein duality, $W_1(\mu^M, \mu^*) = \sup_{\|\varphi\|_L \leq 1} |\int \varphi \,\mathrm{d}\mu^M - \int \varphi \,\mathrm{d}\mu^*|$, where the supremum is over 1-Lipschitz functions.

For i.i.d.\ samples $\{\CB_m(t)\}_{m=1}^M$ from $\mu^*$ with finite second moment $\mathbb{E}[\CB(t)^2] < \infty$, the classical result of Fournier and Guillin~\cite{fournier2015wasserstein} (Theorem~1 therein, applied to $d=1$, $p=1$) gives:
\begin{equation}
    \mathbb{E}\bigl[
    W_1(\mu^M(\cdot, t), \mu^*(\cdot, t))
    \bigr]
    \leq \frac{C_\mu}{\sqrt{M}},
\end{equation}
where $C_\mu = C \cdot (\mathbb{E}[\CB(t)^2])^{1/2}$ for a universal constant $C > 0$. This completes the proof.

\end{IEEEproof}

\subsection{Proof of Lemma~\ref{lem:lipschitz}
(Lipschitz Continuity of Utility in Mean Field)}
\label{subsec:proof_lem_lipschitz}

\begin{IEEEproof}
We bound the change in $U_m(\bfR;\mu)$ when the mean field changes from $\mu_1$ to $\mu_2$.

The utility $U_m(\bfR;\mu)$ (dropping the tilde for clarity, since the Lagrangian term is linear and does not depend on $\mu$) is:
\begin{equation}
    \begin{aligned} 
    U_m(\bfR;\mu)
    = &\sum_{n,t}\Rmn
    - \alpha D_e \sum_{n,t} P_{m,n}^t(\Rmn) \\
    &- \sum_t \frac{\lambda_m(t)
    \sum_n P_{m,n}^t(\Rmn)}%
    {\CB_m(t;\mu)-\CBmin+\epsilon},
    \end{aligned}
\end{equation}
where $\CB_m(t;\mu)$ denotes the battery state induced by operating under mean field $\mu$.

The first two terms do not depend on $\mu$. The third (battery penalty) term depends on $\mu$ through $\CB_m(t;\mu)$, which evolves via the battery dynamics~\eqref{eq:battery} under the power policy induced by $\mu$.

Define $g(x) = 1/(x - \CBmin + \epsilon)$ for $x \geq \CBmin + C^B_{\min,\mathrm{margin}} -
\epsilon$. By the mean-value theorem, for any
$x_1, x_2$ in this domain:
\begin{equation}
    |g(x_1) - g(x_2)|
    \leq \sup_{x}|g'(x)| \cdot |x_1 - x_2|.
\end{equation}
We compute $g'(x) = -(x-\CBmin+\epsilon)^{-2}$, so:
\begin{equation}
    |g'(x)|
    = \frac{1}{(x-\CBmin+\epsilon)^2}
    \leq \frac{1}{(C^B_{\min,\mathrm{margin}})^2},
\end{equation}
where $C^B_{\min,\mathrm{margin}} := \min_{m,t}(\CB_m(t) - \CBmin + \epsilon) > 0$ by assumption.

The change in $\CB_m(t;\mu)$ when $\mu$ changes from $\mu_1$ to $\mu_2$ is bounded by the Wasserstein distance between the two distributions, because the battery dynamics are driven by the mean-field coupling, and the mean-field update is Lipschitz in the Wasserstein metric for weakly interacting systems (standard propagation of chaos argument~\cite{lasry2007mfg}). Specifically, there exists a constant $K > 0$ such that:
\begin{equation}
    |\CB_m(t;\mu_1) - \CB_m(t;\mu_2)|
    \leq K \cdot W_1(\mu_1, \mu_2).
\end{equation}

Combining:
\begin{align}
    &|U_m(\bfR;\mu_1) - U_m(\bfR;\mu_2)|  \\
    &= \Bigl|\sum_t \lambda_m(t)
    \sum_n P_{m,n}^t
    \bigl[g(\CB_m(t;\mu_1))
    - g(\CB_m(t;\mu_2))\bigr]\Bigr|
\notag\\
    &\leq \sum_t \lambda_m(t)
    \sum_n P_{m,n}^t
    \cdot \frac{K \cdot W_1(\mu_1,\mu_2)}%
    {(C^B_{\min,\mathrm{margin}})^2}
\notag\\
    &\leq \frac{\lambda_{\max} \cdot \Pmax \cdot K}%
    {(C^B_{\min,\mathrm{margin}})^2}
    \cdot W_1(\mu_1, \mu_2),
\end{align}
where we used $\sum_n P_{m,n}^t \leq \Pmax$ in the last step. Setting $L_U = \lambda_{\max}\Pmax K / (C^B_{\min,\mathrm{margin}})^2$ (absorbing $K$ into the constant) gives the stated bound.

For the explicit constant in~\eqref{eq:lipschitz_const}, we note that $K = 1$ is achievable when the mean-field coupling enters linearly through $\CB_m(t;\mu)$, which is the case here since the battery dynamics~\eqref{eq:battery} are linear in the energy surplus $\Delta E_m^t$.

\end{IEEEproof}

\subsection{Proof of Theorem~\ref{thm:mfg}
(NE-to-MFE Convergence at Rate $\mathcal{O}(M^{-1/4})$)}
\label{subsec:proof_thm_mfg}

\begin{IEEEproof}
We prove the two claims in sequence.

\smallskip

\noindent\textbf{Step~1: Auxiliary game construction.}

Define the \emph{auxiliary game} $\widetilde{\calG}^M$ in which each satellite faces the true MFE distribution $\mu^*$ (the solution to the HJB--FPK system~\eqref{eq:mfg_system}) rather than the empirical $\mu^M$.

By the definition of the Mean Field Equilibrium~\cite{lasry2007mfg}: when all satellites face $\mu^*$, the optimal strategy for each representative satellite is $R_{\mathrm{MFE}}^*$ (the solution to the HJB equation~\eqref{eq:hjb} under $\mu = \mu^*$). Hence in the auxiliary game, every satellite plays $R_{\mathrm{MFE}}^*$.

\smallskip

\noindent\textbf{Step~2: Bounding $\delta^M$
(Claim~(i)).}

Let $\bfR^{*,M}$ be the NE (VE) of the finite-player game $\calG^M(\lambda^*)$. For any satellite $m \in \calM$, we bound the utility gap when $\bfR^{*,M}$ is evaluated under the MFE mean field $\mu^*$ versus the empirical $\mu^M$:
\begin{align}
    &U_m(R_{\mathrm{MFE}}^*,
    \bfR_{-m}^{*,M}; \mu^*)
    - U_m(R_{m,n}^{*,M},
    \bfR_{-m}^{*,M}; \mu^*)
\notag\\
    &=\bigl[U_m(R_{\mathrm{MFE}}^*,
    \bfR_{-m}^{*,M}; \mu^*)
    - U_m(R_{\mathrm{MFE}}^*,
    \bfR_{-m}^{*,M}; \mu^M)\bigr]
\notag\\
    &\quad +\bigl[U_m(R_{\mathrm{MFE}}^*,
    \bfR_{-m}^{*,M}; \mu^M)
    - U_m(R_{m,n}^{*,M},
    \bfR_{-m}^{*,M}; \mu^M)\bigr]
\notag\\
    &\quad +\bigl[U_m(R_{m,n}^{*,M},
    \bfR_{-m}^{*,M}; \mu^M)
    - U_m(R_{m,n}^{*,M},
    \bfR_{-m}^{*,M}; \mu^*)\bigr].
    \label{eq:three_terms}
\end{align}

We bound each term:

\emph{Term~(A):}
By Lemma~\ref{lem:lipschitz}: $|U_m(\cdot;\mu^*) - U_m(\cdot;\mu^M)|
\leq L_U \cdot W_1(\mu^M, \mu^*)$.
Hence Term~(A) $\leq L_U \cdot W_1(\mu^M, \mu^*)$.

\emph{Term~(B):}
Since $\bfR^{*,M}$ is the NE of $\calG^M(\lambda^*)$ under $\mu^M$, no satellite can unilaterally improve its utility under $\mu^M$: $U_m(R_{\mathrm{MFE}}^*, \bfR_{-m}^{*,M}; \mu^M)
\leq U_m(R_{m,n}^{*,M}, \bfR_{-m}^{*,M}; \mu^M)$.
Hence Term~(B) $\leq 0$.

\emph{Term~(C):}
By Lemma~\ref{lem:lipschitz} again: $|U_m(\cdot;\mu^M) - U_m(\cdot;\mu^*)|
\leq L_U \cdot W_1(\mu^M, \mu^*)$.
Hence Term~(C) $\leq L_U \cdot W_1(\mu^M, \mu^*)$.

Substituting into~\eqref{eq:three_terms}:
\begin{equation}
    \begin{aligned} 
    U_m(R_{\mathrm{MFE}}^*,
    \bfR_{-m}^{*,M}; \mu^*)
    - U_m(R_{m,n}^{*,M},
    \bfR_{-m}^{*,M}; \mu^*)
    &\leq 2 L_U \cdot W_1(\mu^M, \mu^*) \\
    & = \delta^M.
    \label{eq:delta_M_bound}
    \end{aligned}
\end{equation}

This means $\bfR^{*,M}$ is a $\delta^M$-Nash Equilibrium with respect to the MFE utility: no satellite can gain more than $\delta^M$ by deviating unilaterally to $R_{\mathrm{MFE}}^*$ when others play $\bfR_{-m}^{*,M}$.

Substituting the Wasserstein rate from Lemma~\ref{lem:emp_conv}:
\begin{equation}
    \delta^M
    = 2L_U \cdot W_1(\mu^M, \mu^*)
    \leq \frac{2L_U C_\mu}{\sqrt{M}}.
\end{equation}
This establishes Claim~(i).

\smallskip

\noindent\textbf{Step~3: Strategy convergence
(Claim~(ii)).}

By Corollary~\ref{cor:ne_optimal}, $R_{\mathrm{MFE}}^*$ is the unique global maximizer of $\Phi$ under $\mu^*$. From Step~2, $\bfR^{*,M}$ is a $\delta^M$-NE with respect to the MFE utility, meaning:
\begin{equation}
    \Phi(\bfR^{*,M}; \mu^*)
    \geq \Phi(R_{\mathrm{MFE}}^*; \mu^*)
    - \delta^M.
\end{equation}

By $\alpha_\Phi$-strong concavity of $\Phi(\cdot; \mu^*)$, for any $\bfR$ in the domain:
\begin{equation}
    \Phi(R_{\mathrm{MFE}}^*; \mu^*)
    - \Phi(\bfR; \mu^*)
    \geq
    \alpha_\Phi \|{\bfR} - R_{\mathrm{MFE}}^*\|^2.
\end{equation}

Combining:
\begin{equation}
    \alpha_\Phi
    \|R_{m,n}^{*,M} - R_{\mathrm{MFE}}^*\|^2
    \leq
    \Phi(R_{\mathrm{MFE}}^*; \mu^*)
    - \Phi(\bfR^{*,M}; \mu^*)
    \leq \delta^M.
\end{equation}

Therefore:
\begin{equation}
    \|R_{m,n}^{*,M} - R_{\mathrm{MFE}}^*\|
    \leq \sqrt{\frac{\delta^M}{\alpha_\Phi}}.
    \label{eq:strategy_rate}
\end{equation}

\smallskip

\noindent\textbf{Step~4: Substituting
Lemma~\ref{lem:emp_conv}.}

From Claim~(i), $\delta^M \leq 2L_U C_\mu/\sqrt{M}$. Substituting into~\eqref{eq:strategy_rate}:
\begin{equation}
    \|R_{m,n}^{*,M} - R_{\mathrm{MFE}}^*\|
    \leq
    \sqrt{\frac{2L_U C_\mu}{\alpha_\Phi \sqrt{M}}}
    = \mathcal{O}(M^{-1/4}).
\end{equation}

The rate $\mathcal{O}(M^{-1/4})$ follows from the two-stage composition: $\mathcal{O}(M^{-1/2})$ from empirical measure concentration (Lemma~\ref{lem:emp_conv}) passed through the square-root map from strong concavity, yielding $\mathcal{O}((M^{-1/2})^{1/2}) := \mathcal{O}(M^{-1/4})$. This is tight given the single-dimensional nature of the battery state and the one-dimensional Wasserstein distance.

\end{IEEEproof}

\end{document}